\pdfoutput=1

	\documentclass[a4paper,11pt,fleqn,american]{scrartcl}
	\def\version{arxiv}

\usepackage{ifthen}
\def\arxivversion{arxiv}
\def\localversion{local}
\def\algorithmicaversion{algorithmica}
\newcommand\ifarxiv[2][]{\ifthenelse{\equal{\version}{\arxivversion}}{#2}{#1}}
\newcommand\iflocal[2][]{\ifthenelse{\equal{\version}{\localversion}}{#2}{#1}}
\newcommand\ifalgo[2][]{\ifthenelse{\equal{\version}{\algorithmicaversion}}{#2}{#1}}

\usepackage{lmodern}
\usepackage[utf8]{inputenc}
\usepackage[T1]{fontenc}

\usepackage{microtype}

\ifalgo[{%
	\usepackage{colonequals}
}]{%
	\journalname{Algorithmica}
	\usepackage{newtxtext}
	\usepackage[%
		bigdelims,
		varg,
		noamssymbols,
		cmbraces,
		cmintegrals
	]{newtxmath}
	\let\colonequals\coloneqq
	
	\AtBeginDocument{%
	  \DeclareFontShape{\encodingdefault}{\rmdefault}{m}{sl}{<-> ptmro8t}{}%
	  \DeclareFontShape{\encodingdefault}{\rmdefault}{b}{sl}{<-> ptmbo8t}{}%
	  \DeclareFontShape{\encodingdefault}{\rmdefault}{bx}{sl}{<->ssub * ptm/b/sl}{}%
	}
}

\usepackage{marginnote,changepage,relsize,xspace,nameref,enumitem,booktabs,needspace}
\iflocal{
  \usepackage[backgroundcolor=black!7,linecolor=black!50,textsize=footnotesize]{todonotes}
  
	\usepackage[color,notcite,notref]{showkeys}
	
}
\ifarxiv{
  \usepackage[disable]{todonotes}
}
\ifalgo{
  \usepackage[disable]{todonotes}
}

\ifalgo[{%
	\AtBeginDocument{%
		\SetEnumitemKey{bmargin}{%
			leftmargin={\widthof{\bfseries #1}+\labelsep},
			labelwidth={\widthof{\bfseries #1}},
		}%
		\SetEnumitemKey{margin}{%
			leftmargin={\widthof{\bfseries\sffamily #1}+\labelsep},
			labelwidth={\widthof{\bfseries\sffamily #1}},
		}%
	}
}]{%
	\AtBeginDocument{%
		\SetEnumitemKey{bmargin}{%
			leftmargin={\widthof{\bfseries #1}+\labelsep},
			labelwidth={\widthof{\bfseries #1}},
		}%
		\SetEnumitemKey{margin}{%
			leftmargin={\widthof{#1}+\labelsep},
			labelwidth={\widthof{#1}},
		}%
	}
}

\ifalgo[{%
	\setlength{\parskip}{.5\baselineskip}
	\setlength{\parindent}{0pt}
}]{%
	\usepackage{noindentafter}
	\AtBeginDocument{
		\NoIndentAfterEnv{indented}
		\NoIndentAfterEnv{problem}
		\NoIndentAfterEnv{theorem}
		\NoIndentAfterEnv{lemma}
		\NoIndentAfterEnv{proof}
		\NoIndentAfterEnv{corollary}
		\NoIndentAfterEnv{proposition}
		\NoIndentAfterEnv{definition}
		\NoIndentAfterEnv{example}
		\NoIndentAfterEnv{examplectd}
		\NoIndentAfterEnv{algorithm}
	}
}

\ifalgo[{%
	\pagestyle{headings}
	\addtokomafont{caption}{\sffamily\small}
	\addtokomafont{captionlabel}{\sffamily\textbf}
	\setcapmargin{2em}
}]{}

\usepackage{letltxmacro}
\LetLtxMacro{\oldtodo}{\todo}
\renewcommand{\todo}[2][]{\tikzexternaldisable\oldtodo[#1]{#2}\tikzexternalenable}

\usepackage{babel}
\input{ushyphex.tex} 

\makeatletter%
\ifx\version\localversion
	\useshorthands*{"}
	  \defineshorthand{"/}{\mbox{-}\nolinebreak\bbl@allowhyphens}%
	\addto\extrasamerican{\languageshorthands{ngerman}}%
	\expandafter\xspaceaddexceptions\expandafter{%
	  \csname ngerman@sh@"@=@\endcsname
	}
\else
	\global\catcode`"=13
	\global\def"#1{%
		\ifthenelse{\equal{#1}{=}}{-\bbl@allowhyphens}{}%
		\ifthenelse{\equal{#1}{/}}{-\nobreak\bbl@allowhyphens}{}%
		\ifthenelse{\equal{#1}{~}}{-\nobreak}{}%
	}
\fi
\makeatother

\usepackage{amsmath,amssymb,nicefrac,mathtools,bm}
\usepackage{tikz}
\usetikzlibrary{positioning,arrows,decorations.pathreplacing,calc}
\ifarxiv[{%
  \usepackage{pgfplots,pgfplotstable}
  \pgfplotsset{compat=1.12}
}]{
  \newcommand{\pgfplotstableread}[2]{}
}

\pgfdeclarelayer{background}
\pgfdeclarelayer{foreground}
\pgfsetlayers{background,main,foreground}

\usetikzlibrary{external}
\tikzset{
	external/export=true,
	external/mode=list and make
}

\usepackage{csquotes}
\iflocal[{%
	\usepackage[backend=bibtex,bibstyle=alphabetic,citestyle=alphabetic,maxbibnames=99]{biblatex}
}]{%
	\usepackage[backend=biber,bibstyle=alphabetic,citestyle=alphabetic,maxbibnames=99]{biblatex}
}
\renewbibmacro{url}{%
  \iffieldundef{eprint}{\iffieldundef{doi}{%
    \iffieldundef{url}{}{%
      \printfield{url}
    }
  }{}}{}
}

\usepackage{wref}
\ifarxiv[{%
  \usepackage[
	  final,
	  unicode=true, 
	  bookmarks=true,
	  bookmarksnumbered=true,
	  bookmarksopen=true,
	  breaklinks=true,
	  pdfborder={0 0 0}, 
	  colorlinks=false   
  ]{hyperref}
}]{}

\usepackage[amsmath,hyperref,thmmarks]{ntheorem}
\theorembodyfont{\slshape}
\theoremseparator{:}
\makeatletter
\newtheoremstyle{proofstyle}%
  {\item[\theorem@headerfont\hskip\labelsep ##1\theorem@separator]}%
  {\item[\theorem@headerfont\hskip\labelsep ##1 of ##3\theorem@separator]}
\makeatother

\qedsymbol{\raisebox{-.25ex}{$\Box$}}

\vfuzz \hfuzz	

\theoremstyle{break}
\ifarxiv[%
	\theorempreskip{2\topsep}
]%
{%
	\theorempreskipamount2\topsep
}
\newtheorem{theorem}{Theorem}[section]

\theoremstyle{plain}
\ifarxiv[%
	\theorempreskip{\topsep}
]{%
	\theorempreskipamount\topsep
}

\newtheorem{lemma}[theorem]{Lemma}

\newtheorem{corollary}[theorem]{Corollary}
\newtheorem{definition}[theorem]{Definition}

\theoremstyle{plain}
\theorembodyfont{\upshape}

\theorembodyfont{\upshape}
\newtheorem{algorithm}{Algorithm}
\newtheorem{problem}{Problem}
\theorembodyfont{\upshape\small}
\newtheorem{example}[theorem]{Example}
\def\continuedexref{ex:dummyref}
\newtheorem*{examplecontinued}{\wref{\continuedexref}\space Continued}

\newenvironment{examplectd}[1]
  {\def\continuedexref{#1}\examplecontinued}
  {\endexamplecontinued}

\theorembodyfont{\upshape}
\theoremsymbol{\raisebox{-.15ex}{$\Box$}}
\qedsymbol{\raisebox{-.15ex}{$\Box$}}
\theoremstyle{proofstyle}
\newtheorem{proof}{Proof}

\newcommand{\N}{\ensuremath{\mathbb{N}}}
\newcommand{\Q}{\ensuremath{\mathbb{Q}}}

\newcommand{\Oh}{\ensuremath{O}\xspace}
\newcommand{\oh}{\ensuremath{o}\xspace}
\newcommand{\Th}{\ensuremath{\Theta}\xspace}
\newcommand{\Om}{\ensuremath{\Omega}\xspace}

\newcommand\harm[1]{\ensuremath{H_{#1}}}

\newcommand{\mset}[1]{\boldsymbol{\mathbf{#1}}}

\newcommand\ce{\colonequals}

\newcommand\rel[1]{\mathrel{\:{#1}\:}}
\newcommand\wrel[1]{\mathrel{\;{#1}\;}}
\newcommand\wwrel[1]{\mathrel{\;\;{#1}\;\;}}

\newcommand{\cuts}{\ensuremath{c}\xspace}
\newcommand{\lpieces}{\ensuremath{m}\xspace}
\newcommand{\allpieces}{\ensuremath{p}\xspace}
\newcommand{\Cuts}{\cuts}
\newcommand{\Lpieces}{\lpieces}

\newcommand{\Feasible}{\ensuremath{\operatorname{Feasible}}\xspace}

\newcommand{\optlength}{\ensuremath{l^{\star}}\xspace}
\newcommand{\candidates}{\ensuremath{\mathcal{C}}\xspace}
\newcommand{\allcandidates}{\ensuremath{\candidates_{\mathrm{all}}}\xspace}
\newcommand{\candidatesmultiset}{\ensuremath{\mset{\mathcal{C}}}\xspace}
\newcommand{\allcandidatesmultiset}{\ensuremath{\candidatesmultiset_{\mathrm{all}}}\xspace}

\newcommand{\Ico}{\ensuremath{I_{\mathrm{co}}}\xspace}
\newcommand{\Lco}{\ensuremath{L_{\mathrm{co}}}\xspace}

\newcommand{\cutproblemascii}{Envy-Free Fixed-Length Stick Division}
\newcommand{\cutproblem}{\hyperref[prob:cutproblem]{\textsf{\cutproblemascii}}\xspace}
\newcommand{\optproblemascii}{Envy-Free Stick Division\xspace}
\newcommand{\optproblem}{\hyperref[prob:optproblem]{\textsf{\optproblemascii}}\xspace}

\newcommand{\canonicalcutalg}{\textsc{CanonicalCutting}\xspace}
\newcommand{\algascii}{SearchLstar}
\newcommand{\alg}[2][]{%
  \hyperref[alg:searchalg]{%
    \ensuremath{%
      \textsc{\algascii}#1%
      \ifthenelse{\isempty{#2}}{%
      }{% Show parameters
        \langle #2 \rangle%
      }%
    }%
  }%
  \xspace%
}
\newcommand{\coolalgascii}{SelectLstar}
\newcommand{\coolalg}[2][]{%
  \hyperref[alg:selalg]{%
    \ensuremath{%
      \textsc{\coolalgascii}#1%
      \ifthenelse{\isempty{#2}}{%
      }{% Show parameters
        \langle #2 \rangle%
      }%
    }%
  }%
  \xspace%
}

\ifalgo[{%
	\newenvironment{indented}{%
	  \begin{adjustwidth}{2em}{2em}
	}{%
	  \end{adjustwidth}
	}
}]{%
	\newenvironment{indented}{%
		\smallskip
	  \begin{adjustwidth}{2em}{2em}
	}{%
	  \end{adjustwidth}
	  \smallskip
	}
}

\newdimen\makeboxdimen

\newcommand\plaincenter[1]{%
	\mbox{}\hfill#1\hfill\mbox{}%
}

\def\mydots{\xleaders\hbox to.75em{\hfill.\hfill}\hfill}
\newlength\tmpLenNotations

\newenvironment{notations}[1][10em]{%
	\small
	\newcommand\notationentry[1]{%
		\settowidth\tmpLenNotations{##1}%
		\ifthenelse{\lengthtest{\tmpLenNotations > \labelwidth}}{%
			\parbox[b]{\labelwidth}{%
				\makebox[0pt][l]{##1}\\%
			}%
		}{%
			\mbox{##1}%
		}%
		\mydots\relax%
	}%
	\begin{list}{}{%
		\setlength\labelsep{0em}%
		\setlength\labelwidth{#1}%
		\setlength\leftmargin{\labelwidth+\labelsep+1em}%
	}
	\newcommand\notation[1]{\item[{##1}]}
	\iflocal{%
		\renewcommand\notation[1]{\item[{##1}] \marginline{\tiny\texttt{\detokenize{##1}}}}
	}
	\raggedright
}{%
	\end{list}
}

\title{%
	Building Fences Straight and High -- 
	An Optimal Algorithm for Finding the Maximum Length 
	You Can Cut \textit{k} Times from Given Sticks
}

	\author{%
		Raphael Reitzig%
		\footnote{%
			Recreational Researcher, 
			formerly at University of Kaiserslautern,
			\texttt{reitzig@verrech.net}
		} 
	\and 
		Sebastian Wild%
		\footnote{%
			David R.\ Cheriton School of Computer Science, 
			University of Waterloo;
			\texttt{wild@waterloo.ca}
		}
	}
\ifarxiv{%
	\usepackage{hyperref}
}

\AtBeginDocument{%
	\hypersetup{
	  pdftitle={Building Fences Straight and High - An Optimal Algorithm for Finding the Maximum Length You Can Cut k Times from Given Sticks},
	  pdfauthor={Raphael Reitzig and Sebastian Wild},
	  hidelinks=true
	}
}

\begin{document}

\maketitle

\begin{abstract}\noindent
	Given a set of $n$ sticks of various (not necessarily different) lengths,
	what is the largest length so that we can cut $k$ equally long pieces of this length 
	from the given set of sticks?
	%
	We analyze the structure of this problem and show that it essentially reduces
	to a single call of a selection algorithm;
	we thus obtain an optimal linear"=time algorithm.

	This algorithm also solves the related envy-free stick-division problem,
	which \textcite{Segal-Halevi2016} recently used as their central primitive operation for the
	first discrete and bounded envy-free cake cutting protocol with a proportionality
	guarantee when pieces can be put to waste.
	\ifalgo[{%
		\par\smallskip\noindent
		\textbf{Keywords:}\\
			envy-free stick division, envy-free allocations, fair division, 
			building fences, stick cutting, cake cutting with waste, 
			proportional apportionment
	}]{%
		\keywords{%
			envy-free stick division, envy-free allocations, fair division, 
			building fences, stick cutting, cake cutting with waste, 
			proportional apportionment
		}
	}%
\end{abstract}

\section{Introduction}

This article originates from an apparently innocuous problem posed to the 
online question"=and"=answer network \href{https://cs.stackexchange.com}{\textsl{Computer Science Stack Exchange}} 
in September 2014:

\begin{indented}\slshape
	\mbox{\llap{``}You} have $n$ sticks of arbitrary lengths, not necessarily integral.
	By cutting some sticks (one cut cuts one stick, but we can cut as often as we want), 
	you want to get $k<n$ sticks such that:
	\begin{itemize}
		\item All these $k$ sticks have the same length;
		\item All $k$ sticks are at least as long as all other sticks.
	\end{itemize}
	Note that we obtain $n+C$ sticks after performing $C$ cuts.
	
	What algorithm would you use such that the number of necessary cuts is minimal? 
	And what is that number?''
	\upshape\hfill\cite{csse-question}
\end{indented}

Erel Segal-Halevi posed the question because he and his coauthors Avinatan Hassidim and Yonatan Aumann 
used this very procedure as their basic primitive to devise the first discrete and bounded 
envy-free cake cutting protocol for any number of agents, when it is acceptable to leave some
pieces of the cake unassigned (these pieces go to waste)~\cite{mainrefConf,Segal-Halevi2016}.
Their work constituted a significant progress on a long-standing open problem;
and a key technical lemma in their work uses the algorithms devised in this paper.
We give some background on cake cutting and some details about the protocol of 
\citeauthor{Segal-Halevi2016} in \wref{sec:cake-cutting-erel}.

The proposed problem~-- we will call it \optproblem~-- is quite elementary in nature and 
one expects to find an efficient solution using only basic data structures if properly combined~--
the feeling is that something along the lines of a binary search should do the job.
It seemed like an ideal creative exercise problem for students,
and indeed, the authors of this paper posed a simplified version of \optproblem 
as a bonus problem in a written exam for an intermediate-level algorithms course.
While preparing a detailed solution, though,
we found the problem surprisingly intriguing;
in particular, all answers given by the Stack Exchange community at that time 
were either far from optimal or
had significant gaps in the argumentation; a linear-time solution was not even speculated about.
In the end, the first impression turned out right~-- 
there is a simple and elementary algorithm that optimally solves \optproblem~-- 
but finding it was well beyond the scope of a typical exercise problem.
Somewhat unexpectedly, one can entirely avoid sorting the stick lengths and 
use a single call to a (rank) selection algorithm instead, leading to a linear-time algorithm.

The formulation of the problem above asks for the minimal number of cuts $C$, 
but \citeauthor*{Segal-Halevi2016}{} do not use $C$ itself in the end, 
but rather the \emph{length} of the longest sticks.
It is a trivial observation (see \wref{sec:problem-definition}) that the crux of the problem
is indeed finding 
\textsl{the maximal length \(\optlength\) so that \(k\) sticks of that length are 
obtainable by cuts \textit{at all};}
given~\optlength, we can easily determine how often and where to cut sticks to solve the
original \optproblem problem.
If we only ask for $k$ longest equal-length sticks, the second requirement,
that all other sticks must be (weakly) shorter, becomes immaterial.
This makes the problem even shorter to formulate, and quite practical:

\begin{indented}\slshape
	Assume you have a supply of $n$ sticks \textsmaller{(or planks; or poles; \dots)} of various lengths, 
	and you need $k$ equally long pieces \textsmaller{(as legs for a table; as posts for a fence; 
	as boards for a shelf; \dots)}.
	
	What is the maximal length of the pieces you can get from these (without~glue)?
\end{indented}

Despite its natural applications, 
this article is the first algorithmic treatment of the problem to the best of our knowledge.

Interestingly, the \optproblem problem was also part of a recent programming contest:
the ICI Open 2016 organized by 
the Norwegian University of Science and Technology in Trondheim.
Torbj{\o}rn Morland came up with the problem independently of the Stack Exchange question
[private communication] and he formulated it in yet another context
(as cutting equally high posts for a fence).%
\footnote{%
	It is an unlucky coincidence that \citeauthor{Morland2016} used
	the same parameter names $n$ and $k$ with exactly reversed roles \dots
}
The problem was called ``Building Fences'' in the contest and code can still be submitted 
(outside the original competition)
in the Kattis online system~\cite{Morland2016}.
We submitted a straight-forward C++ implementation of the algorithm devised in this paper,
and this submission immediately was the fastest of all 50-odd accepted submissions to date.
We take this anecdote as a sign that
(a) \optproblem is indeed a reasonably natural problem 
(natural enough to be rediscovered independently),
(b) our proposed algorithm is indeed efficient in practice (also for small input sizes),
and (c) that the algorithmic idea does not (yet) seem to be folklore,
making it well worthy of a proper discussion.

There is a less immediate connection to the problem of adequately
assigning seats in parliament to parties according to their relative share of votes after an election:
as we discuss in \wref{sec:relation-to-apportionment},
\optproblem can be formulated as such an apportionment problem, and likewise
the algorithm for cutting sticks we devise below
can be 
transferred and generalized to proportional apportionment.
We show in a companion paper~\cite{ReitzigWild2015} that
the resulting algorithm is indeed an improvement over state-of-the-art methods
for proportional apportionment with divisor sequences: 
it is the first algorithm with both a worst-case guarantee of linear running time,
and practical performance on par with the best heuristic algorithms currently in use.

We conclude this introduction with another metaphor for the problem considered in this paper.
We found this metaphor the most memorable one that includes the requirement of an envy-free
division of the resources; it is however not meant as a serious application.

\begin{indented}\slshape
Imagine you and your family move to a new house.
Naturally, each of your $k$ children wants to have a room of their own,
which is why you wisely opted for a large house with many rooms.
The sizes of the rooms are, however, not equal, and you anticipate 
that peace will not last long if any of the rascals finds out that
their room is smaller than any of the others'.

Removing walls is out of the question, but any room can 
(arbitrarily and iteratively) be
divided into smaller rooms by installing drywalls.
What is the minimal number of walls needed to obtain a configuration
that allows you to let your kids freely choose their rooms with guaranteed envy-freeness?
\end{indented}

We would like point out two differences to typical fair"=division scenarios: 
We assume (somewhat unrealistically for children) that only the size of
the rooms matters, so that two rooms of the same size are essentially indistinguishable.
Moreover, we make the (in our opinion sensible) assumption that all resulting rooms 
can in principle be chosen; this is in contrast to division scenarios with free disposal where 
agents do \emph{not} envy resources that were laid to waste.
Although we like the children's\/"=rooms metaphor
we prefer to remain consistent with existing descriptions of the problem and
will therefore continue talking about \emph{sticks} that are cut instead of rooms with installed walls.

In the remainder of this first section, we give an overview over related problems
and a roadmap of our contribution.
We then introduce notation and give a formal definition of \optproblem
in \wref{sec:problem-definition}. 
In \wref{sec:structure} we develop the means to limit our search
to finite sets of candidates.
We follow up by developing first a linearithmic, then a linear time algorithm
for solving \optproblem in \wref{sec:algorithms}.
Finally, we propose smaller candidate sets in \wref{sec:candidate-sets}.
A short conclusion in \wref{sec:conclusion} on the complexity of \optproblem
completes the article.

We append a glossary of the notation we use in \wref{app:notations} 
for reference, and some lower bounds on the number of distinct candidates
in \wref{app:additional-results}.

\subsection{Other Optimization Goals}
\label{sec:other-optimization-goals}

We briefly discuss of a few variants of stick cutting,
none of which changes the nature of the problem.
To reiterate, we consider the problem of dividing resources (fairly),
some of which may remain unallocated.
These are wasted, which of course is to be avoided if possible.
We can formulate this goal in different ways;
one can seek to
\begin{enumerate}[label=(G\arabic*),leftmargin=3.5em]
	\item \label{goal:min-cuts}
		minimize the number of necessary cuts (as in the original \optproblem),
	\item \label{goal:min-num-waste}
		minimize the number of waste pieces,
	\item \label{goal:min-amount-waste}
		minimize the total amount of waste 
		(here, the total length of all unallocated pieces), 
		or
	\item \label{goal:max-lstar}
		maximize the amount of resource each player gets
		(the length of the maximal pieces). 
\end{enumerate}
The first two objectives are discrete (counting things) whereas the 
latter two consider continuous quantities.

Obviously, \ref{goal:min-amount-waste} and \ref{goal:max-lstar} are dual 
formulations for equivalent problems; 
the total waste is always the (constant) total length of all sticks minus $k \cdot \optlength$.
Similarly, \ref{goal:min-cuts} is dual to \ref{goal:min-num-waste};
$c$ cuts divide $n$ sticks into $n+c$ pieces, and because exactly $k$ of these 
are non"=waste, the number of wasted pieces is $n+c - k$.

Recall that we require also the unassigned sticks to be cut so that they are no longer 
than the $k$ equal sticks.
This implies that the \emph{canonical division} induced by the largest feasible cut length~$l$~--  
cutting length-$l$ pieces off of any longer sticks until no such are left~-- 
is also optimal w.\,r.\,t.\ the number of cuts: 
a smaller cut length can only lead to more cuts.
This is one of the key insights towards algorithmic solutions of \optproblem 
so we state it formally in \wpref{cor:optismaxfeasible}.
In the terms of above goals, this means that optimal solutions for 
\ref{goal:min-amount-waste} and \ref{goal:max-lstar} are also optimal
for \ref{goal:min-cuts} and \ref{goal:min-num-waste}.

The converse is also true. Let \optlength be the cut length of an 
optimal canonical division w.\,r.\,t.\ the number of needed cuts 
(goal \ref{goal:min-cuts}).
This division must divide (at least) one stick perfectly, 
that is into only maximal pieces;
otherwise we could increase \optlength a tiny bit
until one stick is perfectly split, resulting in (at least) one cut less 
(as this stick does not produce a waste piece now).
But this means that we cannot \emph{increase} \optlength 
by any (positive) amount without immediately losing (at least)
one maximal piece; we formalize this fact in \wpref{lem:monotonicity+step}.
As the number of cuts grows with decreasing cut length, 
\optlength is the largest cut length that yields at least $k$ maximal pieces.
Together it follows that \optlength must be the
largest feasible cut length overall and thus induces an optimal 
solution for goal \ref{goal:max-lstar} (and therewith \ref{goal:min-amount-waste}), 
too. 

Therefore, all four objective functions we give above result in the same
optimization problem, and the same algorithmic solutions apply.
The reader may use any of the four formulations in case they do not agree
with our choice of \ref{goal:min-cuts}.

Note that the requirement that unassigned sticks must also be cut is essential for this equivalence.
Minimizing the number of cuts to obtaining any $k$ equal pieces (of arbitrary positive length)~--
or equivalently, producing an envy-free allocation when agents do \emph{not} envy wasted pieces~-- 
is indeed a \emph{different} problem that we do not consider here.%
\footnote{
	As one of our reviewers pointed out, the total fraction of waste is \emph{unbounded} 
	if we minimize cuts in the allocated part only: 
	consider the input $2,1,\ldots,1, M$, where $M$ is a huge number;
	it requires only one cut to assign pieces of length $1$, so an arbitrarily large fraction
	of the resource goes to waste.
	In contrast, when the cuts in wasted pieces also count, we never waste more than half of the 
	total resource (unless $n > k$ already initially).
	
	From an algorithmic perspective, the assumption that waste is not envied means that the algorithm
	must decide which sticks are deliberately put to waste and need hence not be cut at all.
	The number of performed cuts is then no longer a monotonic function of the cut length
	(as it is in our case, see \wref{sec:overview});
	rather it depends on how many sticks can be cut \emph{evenly} into maximal pieces.
	Our efficient solution by selection does hence not solve this problem.
}
\todo{
	Should we? It is also natural, right?
	
	We never need more than $k-1$ cuts for this;
	can only use less when we ``reuse'' existing cuts, i.\,e., we should maximize the number of
	different sticks used in the allocation that are cut evenly.
	This function is no longer monotonic, and certainly does not allow the select-trick.
	
	It requires a quite different approach.
}

\subsection{The Origins: Cake Cutting}
\label{sec:cake-cutting-erel}

\optproblem was motivated by a recent approach to \textsl{envy-free cake cutting.}
We briefly describe the context of the problem and how the stick division problem
is used for cake cutting.

The fair allocation of resources is a well-studied problem in economics.
See, e.\,g., \textcite{Brams1996} for a general treatise of the field and
\textcite{Brandt2016} for recent developments with a focus on computational aspects.
A vital feature of all fair division problems is that valuations are subjective:
different players may assign different values to the same objects,
i.\,e., the same piece of cake in the cake-cutting problem.

The cake-cutting problem has become the predominant mathematical metaphor for allocating
an infinitely divisible, inhomogeneous resource 
``fairly'' to a number of competing players or agents.
The standard model is to consider the unit interval as the cake (a very thin cake, alas).
An agent's value for a piece of cake does not only depend on the size of the piece, 
but also on its position within the cake (say, because the topping of the cake differs).
In general, assigning \emph{disconnected} pieces, i.\,e., a finite union of intervals, 
is acceptable, but sometimes \emph{contiguous} pieces are explicitly required.
Agents are not willing to share pieces, so all assigned pieces must be disjoint.
To exclude degenerate situations, valuations are assumed to be additive and 
absolutely continuous w.\,r.\,t.\ length.

What exactly constitutes a ``fair'' division is subject to debate and
several (partly contradicting) notions of fairness appear in the 
literature:
\begin{itemize}
\item 
	\emph{proportional division} (or \emph{simple fair division}) guarantees
	that each player gets a share that she values at least $1/k$;
\item 
	\emph{envy"=freeness} ensures that no player (strictly) prefers another 
	player's share over her own; 
\item 
	\emph{equitable division} requires that the (relative) value each player
	assigns to her own share is the same for all players, 
	(everyone feels the same amount of ``happiness'').
\end{itemize}
All three notions of fairness have been studied extensively for cake cutting.

Despite the maturity of the field, several ground-breaking results on envy"=free cake cutting
have only been found very recently;
in fact concurrently to the preparation of this article, \textcite{AzizMackenzie2016} published 
the first discrete and bounded protocol to produce an envy"=free allocation for any number of agents,
settling an open problem intensively studied at least since the 1960s.%
\footnote{%
	The authors thank the reviewers for pointing out this recent development.
}
Many solutions for variations and restrictions of this problem have been proposed,
and in view of the tremendous complexity of \citeauthor{AzizMackenzie2016}'s protocol~--
it does not greedily assign pieces to agents, but potentially reallocates them many times~--
those will probably remain relevant in practical applications.

\textcite{mainrefConf,Segal-Halevi2016} devise a much simpler protocol to find an envy"=free
division of a cake to $n$ agents when parts of the cake may remain unassigned.
Their core method (Algorithm~1 in~\cite{Segal-Halevi2016}) 
is a protocol that allocates \emph{contiguous} pieces to $n$ agents
with the guarantee that each agent values her piece at least $1/2^{n-1}$ of the overall cake.
This may seem little, but the generic protocol can be further improved for $n=3$ and $n=4$,
and if it is applied iteratively to non-assigned pieces, it yields an almost proportional
envy"=free division with disconnected pieces.

\Citeauthor*{Segal-Halevi2016} use the algorithm for solving \optproblem presented in this paper 
as a subroutine in their protocol,
namely for implementing the $\mathit{Equalize}(k)$ queries (Lemma~4.1~\cite{Segal-Halevi2016}). 
In short, their core method works as follows:
The players cut pieces from the cake, one after another, 
only allowing to subdivide already produced pieces.
After that, players each choose one of their favorite pieces among the existing pieces
in the \emph{opposite} cutting order, 
i.\,e, the player who cut first is last to choose her piece of cake.
During the cutting phase, agents produce a well-chosen number of pieces 
($2^{n-r-1}+1$ if the agent is the $r$th cutter)
of the cake that they regard to be all of equal value and
all other pieces are made at most that large;
this is exactly the setting of \optproblem.
The number of maximal pieces is chosen such that even after all players 
who precede the given player in ``choosing order'' have taken 
their favorite piece of cake, 
there is at least one of the current player's maximal pieces left for
her to pick, guaranteeing envy-freeness of the overall allocation.

\subsection{Stick Cutting As Fair Division Problem}

The setting of \optproblem shares some features of fair division problems,
but we would like to explicitly list several important differences here.

First, the core assumption for classic fair division problems is the
subjective theory of value, 
which in general forbids an objectively fair division of the resources.
In stick cutting, all players agree in their valuations, so that we can speak of \emph{the} length
of any given stick without loss of generality~-- 
or in the alternative metaphor: all the children value same rooms same: (linearly) by their size.

Second, among the above notions of fairness, proportional divisions do not usually exist 
for cutting sticks,
and equitability coincides with envy"=freeness when players agree in their valuations.
Envy-free assignments of sticks do also not exist in general
unless we allow leaving some pieces unallocated.
Note that we use the term envy"=free to imply that also these non-allocated pieces
have to be made unattractive by cutting them
(cf.\ the discussion in \wref{sec:other-optimization-goals}),
whereas in fair allocation scenarios agents do not envy waste.

Third, while each given stick is assumed to be continuously divisible,
existing cuts constrain the set of possible allocations,
so we neither have a purely divisible nor a purely indivisible allocation scenario.

In summary, we think that viewing \optproblem as restrictive special case of fair division
misses the problem's immediate own applications.

\subsection{Implications for Proportional Apportionment}
\label{sec:relation-to-apportionment}

It may be a surprising connection at first sight, 
but the stick cutting problem can be viewed as an apportionment problem in disguise.
Proportional apportionment is the problem
of assigning each party its ``fair'' share of a pool of \emph{indivisible}, 
unit-value resource items (seats in parliament), so
that the fraction of items assigned to a party resembles
as closely as possible
its (fractional, a priori known) \emph{value} 
(the number of votes for a party); 
these values are the input.%
\footnote{%
	We thank Chao Xu for bringing this problem to our attention.
}
\Textcite{Balinski2001} describe the problem extensively and 
illustrate many pitfalls with examples from the rich history of representation systems in the US.
\Textcite{Pukelsheim2014} adds more recent developments and the European perspective; 
he also takes a more algorithmic point of view on apportionment.

Most methods used in real-word election systems use sequential
assignment, allocating one seat at a time in a greedy fashion,
where the current priority of each party depends on its initial vote percentage
and the number of already assigned seats.
Different systems differ in the function for computing these updated priorities,
but all sensible ones are of the ``highest averages form'': 
In each round they assign the next seat to a party that (currently) maximizes 
$\nicefrac{v_i}{d_j}$ (with ties broken arbitrarily), 
where $v_i$ is the vote percentage of party $i$, 
$j$ is the number of seats party $i$ has already been assigned 
and $d_0, d_1, d_2,\ldots$ is an increasing \emph{divisor sequence} 
characterizing the method.%
\footnote{%
	$d_0 = 0$ is allowed in which case $\nicefrac{v_i}{d_0}$ is supposed to mean $v_i + M$ for
	a very large constant~$M$ (larger than the sum $\sum v_i$ of all values is sufficient). 
	This ensures that before any party is assigned a 
	second seat, all other parties must have one seat, 
	which is a natural requirement for some allocations scenarios, 
	e.\,g., the number of representatives for each state in a federal republic 
	might be chosen to resemble population counts, but any state, 
	regardless how small, should have at least one representative.
}

Even though the original description of the highest averages allocation
procedure is an iterative process, it is actually a static problem:
The averages $\nicefrac{v_i}{d_j}$ are strictly decreasing with $j$ for any $i$, so 
the sequence of maximal averages in the assignment rounds is also decreasing. 
In fact, if we allocate a seat to a party with current average $a$ in round $r$, 
then $a$ must have been the $r$th~largest element in the multiset of all 
ratios $\nicefrac{v_i}{d_j}$ for all parties $i$ and numbers $j$.
Moreover, if we know the value of the $k$th largest average $a^{*}$ up front,%
\footnote{%
	\Textcite{Pukelsheim2014} calls this value the divisor $D$, and thus refers 
	to highest-averages methods simply as divisor methods.
}
we can directly determine for each party $i$ how many seats it should receive: 
if $j$ is the largest number such that $\nicefrac{v_i}{d_j} \ge a^{*}$, 
then party $i$ receives $j+1$ seats.

The arguably most natural choice is hence $d_j = j+1$, yielding
the highest average method of Jefferson, a.\,k.\,a.\ the greatest divisors method.
For $d_j = j+1$, a party gets one seat for each time it can afford to pay the 
full price of a seat (namely $a^*$ votes), 
so it is assigned $\lfloor \nicefrac{v_i}{a^*}\rfloor$ seats.
The crux of the problem is thus to find a value $a^*$, such that this rule assigns exactly $k$
seats in total (or the smallest number no less than $k$ in case of ties).
Since $\lfloor \nicefrac{v_i}{a^*}\rfloor$ is also the number of maximal pieces that can be cut
from a stick of total length $L=v_i$ using cut length $l=a^*$,
we are actually asking for a maximal cut length $a^*$ so that we obtain $k$ equally long pieces
when trimming sticks of initial lengths $v_1,\ldots,v_n$ to length $a^*$.
Therefore \optproblem is essentially equivalent to an apportionment problem 
with divisor sequence $d_j=j+1$ and the stick lengths as vote tallies.
Note that for this equivalence, we have to \emph{reverse} the roles of agents and resources: 
\textsl{Assigning equally long stick pieces to players is equivalent to 
apportioning \textit{to the sticks} their fair share of players 
using Jefferson's highest averages method.}

We can consequently use any apportionment algorithm to solve \optproblem,
in particular the practically efficient methods proposed by \textcite{Pukelsheim2014}
or the (rather complicated) worst-case linear-time algorithm of \textcite{Cheng2014}.
However, it actually turns out more fruitful to \emph{reverse} the idea:
our algorithm for \optproblem is conceptually much simpler than the method of \citeauthor{Cheng2014}, 
but has the same time worst"=case linear"=time guarantee, 
so we can actually improve the state-of-the-art methods for apportionment.

The method for \optproblem per se can only deal with the divisor sequence $d_j=j+1$,
but we show in a companion article~\cite{ReitzigWild2015} how to generalize the underlying ideas
to all divisor sequences listed by
\textcite[Table\,1]{Cheng2014}, 
and indeed to any divisor sequence that \citeauthor{Cheng2014}'s algorithm can handle.
Although the techniques remain similar, the generalization required modifications to the formalism,
and it deemed us best to separate the detailed discussion of apportionment
from the stick cutting algorithms in this paper.

In the apportionment article~\cite{ReitzigWild2015}, 
we demonstrate in extensive running"=time experiments that our method is indeed 
much faster than \citeauthor{Cheng2014}'s algorithm, 
and has more predictable running time than the methods currently used in practice.
The latter do not have a linear"=time guarantee in the worst case, 
and we identify a class of inputs, where they indeed exhibit superlinear behavior.

It is somewhat surprising to us (in hindsight) that we did not find our method in the
quite extensive literature on proportional apportionment;
in any case, the detour through \optproblem has helped us in finding it a lot.

\subsection{Overview of this Article}
\label{sec:overview}

In this section, we give an informal description of the steps that lead to our 
solution for \optproblem (see also \wref{fig:schematic-overview});
formal definitions and proofs follow in the main part of the paper.

\begin{figure}
	\plaincenter{
	  \begin{tikzpicture}[auto,Step/.style={draw,rectangle,minimum width=5cm,
	                                        inner sep=5pt},
	                           Major/.style={semithick,fill=gray!10},
	                           Insight/.style={outer xsep=5pt},
	                           rounded corners,
	                           scale=0.9,transform shape]
	    \node[Step,Major]      (s1)  {\optproblem};
	    \node[Step,below=1.5 of s1] (s2)  {one"=dimensional problem};
	    \node[Step,below=1.5 of s2] (s3)  {search problem};
	    \node[Step,below=1.5 of s3] (s4)  {discrete search problem};
	    \node[Step,below=1.5 of s4] (s5)  {finite search problem};
	    \node[Step,below left =1.5 and -1 of s5] (s6a) {linearithmic search space};
	    \node[Step,Major,below right=1.5 and -1 of s5] (s6b) {linear search space};
	  
	    \path[-stealth']
	      (s1) edge node[Insight] {only canonical divisions} (s2)
	      (s2) edge node[Insight] {monotonic objective} (s3)
	      (s3) edge node[Insight] {piecewise constant constraint} (s4)
	      (s4) edge node[Insight] {trivial lower bound on \optlength} (s5)
	      (s5) edge[bend right=10] 
	            node[Insight,swap] {combinatorics} 
	           (s6a)
	      (s5) edge[bend left=10]  
	             node[Insight] {sandwich bounds} 
	           (s6b);
	  \end{tikzpicture}
	}
	\caption{%
		Schematic overview of the refinement steps that turn
		a seemingly hard problem into a tame task amenable to elementary 
		yet efficient algorithmic solutions.
}
	\label{fig:schematic-overview}
\end{figure}
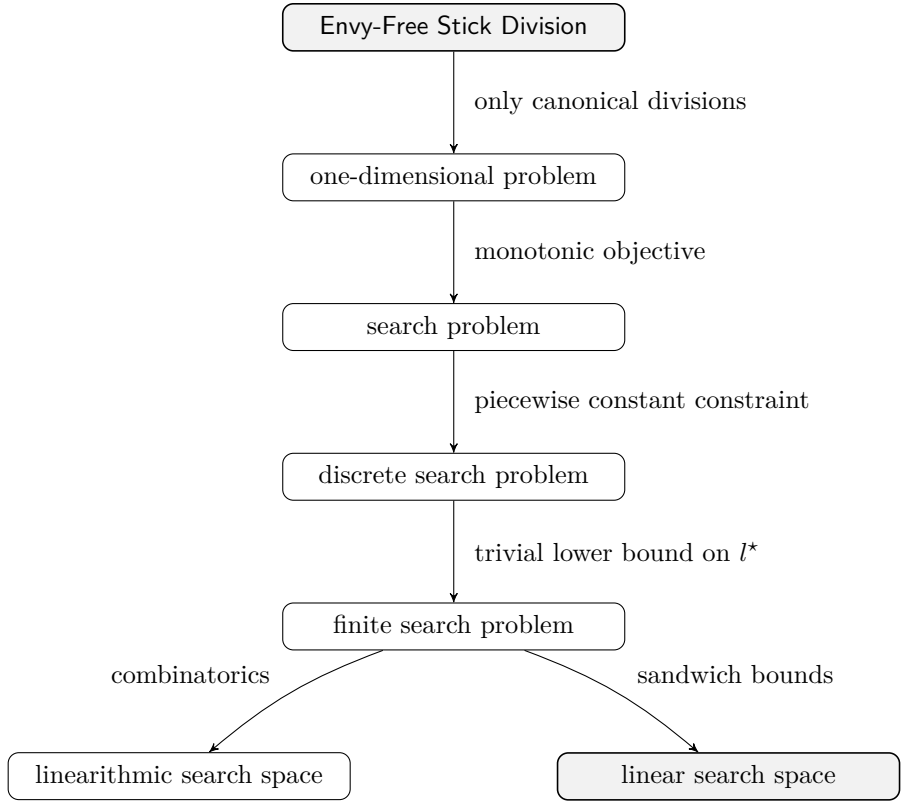

Without further restrictions, 
\optproblem is a non"=linear continuous optimization problem
that does not seem to fall into any of the usual categories of problems
that are easy to solve.
Any stick might be cut an arbitrary number of times at arbitrary lengths,
so the space of possible divisions is huge.

The first step to tame the problem is to observe that most of these divisions
cannot be optimal:
Assuming we already know the size \optlength of the $k$ maximal pieces in 
an optimal division (the size of the rooms assigned to the kids), 
we can recover a \textsl{canonical} optimal division by simply cutting
\optlength{}"/sized pieces off of any stick longer than \optlength 
until all sticks have length at most \optlength.
Cutting a shorter piece only creates waste, cutting a larger piece 
always entails a second cut for that piece.
We can thus identify a (candidate) cut length with its 
corresponding canonical division and so
\optproblem reduces to finding the optimal cut length \optlength.

The second major simplification comes from the observation that
for canonical divisions, the number of cuttings can only get larger
when we decrease the cut length.
(We cut each sticks into shorter pieces, this can only mean more cuts.)
Stated differently, the objective function that we try to minimize is
\emph{monotonic} (in the cut length).
This is a very fortunate situation since it allows a simple characterization
of optimal solutions:
\optlength~is the \emph{largest} length, whose canonical division still contains
(at least) $k$ maximal pieces of equal size,
transforming our optimization to a mere search problem for the point
where cut lengths transition from feasible to infeasible solutions.

By similar arguments, also the number of equal sized maximal pieces
(in the canonical division) for a cut length $l$ does only increase when 
$l$ is made smaller,
so we can use \textsl{binary search} to find the length \optlength 
where the number of maximal pieces first exceeds $k$.
The search is still over a continuous region, though.

Next we note that both the objective and the feasibility function 
are piecewise constant with jumps only at lengths of the form $\nicefrac{L_i}j$, 
where $L_i$ is the length of an input stick and $j$ is a natural number.
Any (canonical) division for length $l$ that does not cut any 
stick evenly into pieces of length $l$
remains of same quality and cost if we change the $l$ a very little.
Moreover, any such division can obviously be improved by increasing 
the cut length, until we cut one stick $L_i$ evenly, say into $j$ pieces,
as we then get the same number of maximal pieces with (at least)
one less cutting.
We can thus restrict our (binary) search for \optlength to these jump points,
making the problem discrete~-- but still infinite, as we do not yet have an upper
bound on $j$.

We can, however, easily find lower bounds on \optlength~--
or, equivalently, upper bounds on~$j$~-- that render the search space finite.
For example, we obviously never need to cut more than $k$ pieces
out of any single stick, in particular not the largest one.
This trivial observation already reduces the search space to $\Oh(n^2)$ candidates,
where $n$ is the number of sticks in the input.

We will then show how to obtain even smaller candidate sets by developing slightly 
cleverer upper and lower bounds for the number of maximal pieces
(in the canonical divisions) for cut length $l$.
The intuitive idea is as follows. 
If we had a single stick with the total length of all sticks,
dividing it into $k$ equal pieces would give us the ultimately
efficient division without any waste.
The corresponding ``ultimate cut length'' is of course easy to compute,
but with pre"=cut sticks, it will usually not be feasible.

However, we know how much the precut sticks can possibly cost us relative to the ultimate division:
each input stick contributes (at most) one piece of waste.
Now imagine the sticks arranged in line, 
so that they form a single long stick with some existing fractures.
When cutting this stick evenly into $n+k$ (not only $k$) pieces, the existing cuts lie in at most
$n$ of these pieces, leaving $k$ segments intact.
Therefore the total length divided by $n+k$ is always a feasible cut length.
With a little diligence (see \wref{sec:linear-candidate-set}), 
we can show that the number of jumps between these ``sandwich bounds'' for \optlength, 
i.\,e., the number of cut lengths to check, remains linear in $n$.
For $k \leq n$, we get an $\Oh(k)$ bound by first removing sticks shorter than the $k$th~largest one.

\smallskip

The discussion above takes the point of view of mathematical optimization,
describing how to reduce the number of candidate cut lengths we have to check;
we are still one step away from turning this into an actual, executable algorithm.
After reducing the problem to a finite search problem,
binary search naturally comes to mind;
we work out the details in \wref{sec:algorithms}.
However, \emph{sorting} the candidate set and \emph{checking feasibility} of 
candidates dominate the runtime of this binary"=search"=based algorithm~-- 
this is unsatisfactory.

As hinted at above,
it is possible to determine \optlength more directly, namely as a specific order 
statistic of the candidate set.
From the point of view of objective and feasibility functions, this trick works
because both functions essentially \emph{count} the number of unit jumps (i.\,e.\ 
occurrences of $\nicefrac{L_i}{j}$) at points larger than the given length.
This approach yields a simple linear"=time algorithm based on a single rank selection; 
we describe it in detail in \wref{sec:selection-algorithm}.

\section{Problem Definition}
\label{sec:problem-definition}

We will mostly use the usual zoo of mathematical notation (as used in theoretical
computer science, that is); see \wref{app:notations} for a comprehensive list.
Since they are less often used, let us quickly introduce
notation for \emph{multisets}, though. For some set $X$, denote a multiset
over $X$ by
  \[ \mset{A} \wwrel= \{ x_1, x_2, \dots \} \]
with $x_i \in X$ for all $i \in [1..|\mset{A}|]$. Note the bold letter; we will
use these for multisets, and regular letters for sets. Furthermore, denote
the \emph{multiplicity} of some $x \in X$ in $\mset{A}$ as $\mset{A}(x)$;
in particular,
  \[ |\mset{A}| \wwrel= \sum_{x \,\in\, X} \mset{A}(x). \]
When we use a multiset as operator range, we want to consider every occurrence 
of $x \in \mset{A}$; for example,
  \[  \sum_{x \,\in\, \mset{A}} f(x) 
    \wwrel= \sum_{i \in [1..|\mset{A}|]} f(x_i)
    \wwrel= \sum_{x \,\in\, X} \mset{A}(x) \cdot f(x). \]
  
As for multiset operations, we use \emph{multiset union} $\uplus$ that adds up 
cardinalities; that is, if $\mset{C} = \mset{A} \uplus \mset{B}$ then 
$\mset{C}(x) = \mset{A}(x) + \mset{B}(x)$ for all $x \in X$.
\emph{Multiset difference} works in the reverse way; if 
$\mset{C} = \mset{A} \setminus \mset{B}$ then 
$\mset{C}(x) = \max \{0, \mset{A}(x) - \mset{B}(x) \}$ for all $x \in X$.

Intersection with a set $B \subseteq X$ is to be read as natural extension
for the usual set intersection; that is, if $\mset{C} = \mset{A} \cap B$ then
$\mset{C}(x) = \mset{A}(x) \cdot B(x)$ for all $x \in X$ 
(we also use the multiplicity notation for ordinary sets, where $B(x) \in \{0,1\}$).

Now for problem"=specific notation.
We call any length $L \in \Q$ a \emph{stick}.%
\footnote{%
	We restrict the input lengths to rational numbers to simplify the presentation;
	all algorithms would work the same in a model that works with exact real numbers,
	and they are numerically stable when using floating-point numbers.
	Our analyses count the number of arithmetic operations and comparisons,
	and apply to any such model.
}
\emph{Cutting} $L$ with length 
$0<l<L$ creates two pieces with lengths $l$ and $L - l$ respectively. 
By iteratively cutting sticks and pieces thereof into smaller pieces,
we can transform a set of sticks into a set of pieces.

We define the following trivial problem for fixing notation.
\begin{problem} \textbf{\cutproblemascii}
  \pdfbookmark[2]{\cutproblemascii}{prob:cutproblem}%
  \label{prob:cutproblem}%
  \begin{indented}
    \textbf{Input:} Multiset $\mset L = \{L_1, \dots, L_n\}$ of sticks with 
      lengths $L_i \in \Q_{> 0}$,
      target number $k \in \N_{>0}$ and cut length $l \in \Q_{> 0}$.
    
    \textbf{Output:} The (minimal) number of cuts necessary for
      cutting the input sticks into sticks $L'_1, \dots, L'_{n'} \in \Q_{> 0}$
      so that
      \begin{enumerate}[label=\roman*)]
        \item \label{item:k-sticks}%
          (at least) $k$ pieces have length $l$, i.\,e.\ $|\{ i \mid L'_i = l \}| \geq k$,
        \item \label{item:all-smaller}%
          and no piece is longer than $l$, i.\,e.\ $L'_i \le l$ for all $i$.
      \end{enumerate}
  \end{indented}
\end{problem}
The solution is immediate; we state it below to demonstrate the notation introduced in the following.

We denote by $\lpieces(L,l)$ the number of stick pieces of length $l$~-- we will 
also call these \emph{maximal} pieces~-- 
we can get when we cut stick $L$ into pieces no longer than $l$.
This is to mean that you first cut $L$ into two pieces, then possibly 
further cut those pieces and so on, until all pieces have length at most $l$.
Obviously, the best thing to do is to only ever cut with length $l$.
We thus have
  \[ \lpieces(L, l) \wwrel= \Biggl\lfloor \frac{L}{l} \Biggr\rfloor . \]
Because we may also produce one shorter piece, the total number of pieces we 
obtain by this process is given by
  \[ \allpieces(L, l) \wwrel= \Biggl\lceil \frac{L}{l} \Biggr\rceil , \]
and 
  \[ \cuts(L, l) \wwrel= \Biggl\lceil \frac{L}{l} \Biggr\rceil - 1  \]
denotes the number of cuts we perform.

We extend this notation to multisets of sticks, that is
\begin{align*}
  \Lpieces(\mset L, l) 
    &\wwrel\ce 
	    \sum_{L \,\in\, \mset{L}} \lpieces(L,l)
    \wwrel=
	    \sum_{L \,\in\, \mset{L}} \Biggl\lfloor \frac{L}{l} \Biggr\rfloor\qquad\text{and} \\[1ex]
  \Cuts(\mset L, l) 
	  &\wwrel\ce 
		  \sum_{L \,\in\, \mset{L}} \Cuts(L_i,l)
	  \wwrel=
		  \sum_{L \,\in\, \mset{L}} \Biggl\lceil \frac{L}{l} - 1 \Biggr\rceil.
\end{align*}
See \wref{fig:visual-notation} for a small example.
\begin{figure}
  \begin{center}
    \begin{tikzpicture}[x=7mm,y=5mm,auto,scale=1.2,transform shape]
      \fill[black!10]
        \foreach \y in {1,2,3,4} {%
          \foreach \x in {1,2.5,4,5.5} {%
             (\x,\y) rectangle +(1,1)
          }
        }
        \foreach \x/\y in {1/5, 7/1, 7/2, 7/3} {%
          (\x, \y) rectangle +(1,1)
        };        
%
      \draw[dashed,black!85]
        \foreach \y/\w in {1/8,2/8,3/8,4/5.5,5/2} {%
           (1,\y+1) -- (\w,\y+1)
        };
%
      \foreach \i/\x/\h in {1/1/5.5, 2/2.5/4.8, 3/4/4.5, 4/5.5/4, 5/7/3.5} {%
        \draw (\x,1) rectangle +(1,\h);
        \node[scale=0.85] at (\x + 0.5,0) {$L_\i$};
      }
%
      \draw[decoration={brace},decorate]
        (0.75,1) -- +(0,1) 
          node[midway,scale=0.85,outer xsep=1mm] {$l$};
    \end{tikzpicture}
  \end{center}
  \caption{%
  	Sticks $\mset L = \{L_1, \dots, L_5\}$ cut with some length $l$.
    Note how $\Lpieces(\mset L, l) = 20$ and $\Cuts(\mset L, l)$ = 19.
    There are four non-maximal pieces.%
   }
  \label{fig:visual-notation}
\end{figure}

Using this notation, the conditions of \cutproblem translate into checking 
whether $\Lpieces(\mset L, l) \ge k$ for cut length $l$; we call such 
$l$ \emph{feasible} cut lengths (for $\mset L$ and $k$).
We define the following predicate as a shorthand:
\[
		\Feasible(\mset L, k, l) 
	\wwrel\ce
		\begin{cases}
			1, & \Lpieces(\mset L, l) \geq k; \\
			0, & \text{otherwise}.
		\end{cases} 
\]
Now we can give a concise algorithm for solving \cutproblem.

\begin{algorithm}{$\canonicalcutalg(\mset L, k, l) :$}%
  \pdfbookmark[2]{Algorithm \canonicalcutalg}{alg:cuttingalg}%
  \label{alg:cuttingalg}
  \begin{enumerate}
    \item If $\Feasible(\mset L, k, l)$:
      \begin{enumerate}[label=\theenumi.\arabic*.]
        \item Answer $\Cuts(\mset L, l)$.
      \end{enumerate}
    \item Otherwise:
      \begin{enumerate}[label=\theenumi.\arabic*.]
        \item Answer $\infty$ (i.\,e.\ ``not possible'').        
      \end{enumerate}
  \end{enumerate}
\end{algorithm}

Assuming the unit"=cost RAM model~-- which we will do in this article~-- 
the runtime of \canonicalcutalg is clearly in $\Oh(n)$; 
evaluation of \Feasible and \Cuts in time $\Oh(n)$ each dominates.
We will see later that a better bound is $\Th(\min(k,n))$ 
(cf.\ \wref{lem:feasible-computation}).

Of course, different cut lengths $l$ cause different numbers of cuts.
We want to find an \emph{optimal} cut length, that is a length $\optlength$ 
which \emph{minimizes} the number of cuts necessary to fulfill conditions~\ref{item:k-sticks} 
and~\ref{item:all-smaller} of \cutproblem. We formalize this as follows.

\begin{problem} \textbf{\optproblemascii}
  \pdfbookmark[2]{\optproblemascii}{prob:optproblem}%
  \label{prob:optproblem}%
  \begin{indented}%
    \textbf{Input:} Multiset $\mset L = \{L_1, \dots, L_n\}$ of sticks with 
      lengths $L_i \in \Q_{> 0}$
      and target number $k \in \N_{>0}$.

    \textbf{Output:} A (feasible) cut length $\optlength \in \Q_{>0}$ which 
      minimizes the result of \cutproblem for $\mset L$, $k$ and $\optlength$.
  \end{indented}
\end{problem}

We observe that the problem is not as easy as picking the smallest $L_i$, cutting
the longest stick into $k$ pieces, or using the $k$th~longest stick (if $k \leq n$).
Consider the following, admittedly artificial example which debunks such simplistic attempts.

\begin{example}\label{ex:nonaivealg}
  Let 
    \[ \mset L = \{mx, (m-1)x + 1, (m-2) \cdot x + 2, 
                    \dots, 
                  \nicefrac{m}{2} \cdot x + \nicefrac{m}{2},
                  x-1, x-1, 
                    \dots \} \; \]
  for a total of 
    $n = m^2$ 
  elements and 
    $k = \nicefrac{3}{8} \cdot m^2 + \nicefrac{3}{4} \cdot m$,
  where $m \in 4\N_{>0}$ and $x > \nicefrac{m}{2}$.
  
  Note that $\optlength = x$, that is in particular
  \begin{itemize}
    \item $\optlength \neq L_i$ and
    \item $\optlength \neq \nicefrac{L_i}{k}$
  \end{itemize}
  for all $i \in [1..n]$. In fact, by controlling $x$ we get an (all but) arbitrary
  fraction of an $L_i$ for \optlength. It is possible to extend the example so
  that ``$mx$''~-- the stick whose fraction is optimal~-- has (almost) 
  arbitrary index $i$, too.
  
  As running example we will use $(\mset{L}_{\mathrm{ex}}, k)$ as defined by $m=4$ and $x=2$,
  that is
  \begin{itemize}
    \item $\mset{L}_{\mathrm{ex}} = \{8, 7, 6, 1, 1, 1, 1, 1, 1, 1, 1, 1, 1, 1, 1, 1\}$ and 
    \item $k = 9$.
  \end{itemize}
  Note that $\optlength = 2$ and $\Lpieces(\mset{L}_{\mathrm{ex}}, \optlength) = 10 > k$ here.
  See \wref{fig:monotonicity+step-example} for a plot of $\Lpieces(\mset{L}_{\mathrm{ex}}, l)$.
\end{example}

\section{Exploiting Structure}
\label{sec:structure}

For ease of notation, we will from now on assume arbitrary but fixed input
$(\mset L, k)$ be given implicitly. In particular, we will use $\Lpieces(l)$
as short form of $\Lpieces(\mset L, l)$, and similar for~\Cuts and \Feasible.

At first, we observe that both constraint and objective function of \optproblem
belong to a specific, simple class of functions.
\begin{figure}
	\pgfplotstableread{
	l	1/l	m(l)	c(l)	
	1.	1.	34.	18.	
	1.14286	0.875	18.	17.	
	1.16667	0.857143	17.	16.	
	1.2	0.833333	16.	15.	
	1.33333	0.75	15.	14.	
	1.4	0.714286	14.	13.	
	1.5	0.666667	13.	12.	
	1.6	0.625	12.	11.	
	1.75	0.571429	11.	10.	
	2.	0.5	10.	8.	
	2.33333	0.428571	8.	7.	
	2.66667	0.375	7.	6.	
	3.	0.333333	6.	5.	
	3.5	0.285714	5.	4.	
	4.	0.25	4.	3.	
	6.	0.166667	3.	2.	
	7.	0.142857	2.	1.	
	8.	0.125	1.	0.	
}\plotpointsLpiecesMfourXtwo
\pgfplotstableread{
	l	1/l	m(l)	c(l)	
	1.	1.	56.	40.	
	1.09091	0.916667	40.	39.	
	1.11111	0.9	39.	38.	
	1.14286	0.875	38.	37.	
	1.2	0.833333	37.	36.	
	1.25	0.8	36.	35.	
	1.33333	0.75	35.	33.	
	1.42857	0.7	33.	32.	
	1.5	0.666667	32.	31.	
	1.6	0.625	31.	30.	
	1.66667	0.6	30.	29.	
	1.71429	0.583333	29.	28.	
	2.	0.5	28.	12.	
	2.4	0.416667	12.	11.	
	2.5	0.4	11.	10.	
	2.66667	0.375	10.	9.	
	3.	0.333333	9.	8.	
	3.33333	0.3	8.	7.	
	4.	0.25	7.	5.	
	5.	0.2	5.	4.	
	6.	0.166667	4.	3.	
	8.	0.125	3.	2.	
	10.	0.1	2.	1.	
	12.	0.0833333	1.	0.	
}\plotpointsLpiecesMfourXthree
\plaincenter{%
\begin{tikzpicture}[
	every node/.style={font={\footnotesize}}
]
\begin{axis}[%
	xmin=1.3,
	xmax=9,
	xlabel={$l$},
	ylabel=$\Lpieces(l)$,
]
	\addplot [
	    thick,
	    jump mark right,
	    mark=*,
	] table [x=l,y=m(l)] {\plotpointsLpiecesMfourXtwo};
	\addplot[only marks,mark=*,mark options={fill=white},opacity=0] 
		 table [x=l,y=c(l)] {\plotpointsLpiecesMfourXtwo};
	
	\addplot[red] coordinates { (0,9) (100,9) };
	\node[above,red] at (axis cs: 6,9) {$k=9$};
	
	\draw[thin, densely dotted] (axis cs:2,10) -- (axis cs:2,-1) ;
	\node[right] at (axis cs:2,1) {$\optlength=2$} ;
	
\end{axis}%
\end{tikzpicture}%
%
%
%
}
	\caption{%
		The number of maximal pieces $\Lpieces(\mset{L}_{\mathrm{ex}}, l)$ in 
		cut length~$l$ for $\mset{L}_{\mathrm{ex}}$ as defined in \wref{ex:nonaivealg}. 
    The filled circles indicate the value of $\Lpieces(\mset{L}_{\mathrm{ex}}, l)$ 
    at the jump discontinuities.
	}
	\label{fig:monotonicity+step-example}
\end{figure}

\begin{lemma}\label{lem:monotonicity+step}
  Functions \Lpieces and \Cuts are non"=increasing, piecewise"=constant functions 
  in $l$ with jump discontinuities of (only) the form $\nicefrac{L_i}{j}$ for 
  $i \in [1..n]$ and $j \in \N_{>0}$.
  
  Furthermore, \Lpieces is left- and \Cuts is right"=continuous.
\end{lemma}

\begin{proof}
  The functions are given as finite sums of terms that are either of the form
  $\lfloor \frac{L}{l} \rfloor$ or 
  $\lceil \frac{L}{l} - 1 \rceil$.
  Hence, all summands are piecewise constant and never increase with growing~$l$.
  Thus, the sum is also a non"=increasing piecewise"=constant function.
  
  The form $\nicefrac{L_i}{j}$ of the jump discontinuities is apparent for each 
  summand individually, and they carry over to the sums by monotonicity.
  
  The missing continuity properties of \Lpieces resp.\ \Cuts follow from
  right"=continuity of $\lfloor \cdot \rfloor$ resp.\ left"=continuity of 
  $\lceil \cdot \rceil$; the direction gets turned around because we consider
  $l^{-1}$ but other than that arithmetic operations maintain continuity.
\end{proof}

See \wref{fig:monotonicity+step-example} for an illustrating plot.

Knowing this, we immediately get lots of structure in our solution space which
we will utilize thoroughly.

\begin{corollary}\label{cor:optismaxfeasible} 
  $ \optlength = \max \{ l \in \Q_{>0} \mid \Feasible(l) \}. $
\qed\end{corollary}
Note in particular that the maximum exists because \Feasible is left"=continuous.

This already tells us that any feasible length gives a lower bound on \optlength.
One particular simple case is $k < n$ since then the $k$th~largest stick is always 
feasible. This allows us to get rid of all properly smaller input sticks, too, 
since they are certainly waste when cutting with any optimal length. As a consequence, 
having \emph{any} non"=trivial lower bound on \optlength already speeds up our 
search by ways of speeding up feasibility checks.
\begin{lemma}\label{lem:feasible-computation}
  Let $L \in \Q_{\geq 0}$ fixed and denote with
    \[     I_{> L} 
       \wwrel\ce \{ i \in [1..n] \mid L_i > L \} \]
  the (index) set of all sticks in $\mset L$ that are longer than $L$. 
  Then,
    \[ \Lpieces(l) \wwrel= \sum_{i \,\in\, I_{> L}} m(L_i, l) \]
  for all $l > L$.
\end{lemma}
\begin{proof}
  Clearly, all summands $\lfloor \nicefrac{L_i}{l} \rfloor$ in the definition
  of $\Lpieces(l)$ are zero for $L_i \leq L < l$.
\end{proof}
As a direct consequence, we can push the time for checking feasibility of a 
candidate solution from being proportional to $n$ down to being proportional 
to the number of $L_i$ larger than a lower bound $L$ on the 
optimal length; we simply preprocess $\mset{L}_{> L}$ in time $\Th(n)$.
Since it is easy to find an $L_i$ that can serve as $L$~-- e.\,g.\ 
any one that is shorter than any known feasible solution~-- we will make use
of this in the definition of our set of candidate cut lengths. 

In addition,  the special shape of \Cuts and \Feasible comes in handy. 
Recall that both functions are step functions with (potential) jump discontinuities 
at lengths of the form $l = \nicefrac{L_i}{j}$ (cf.\ \wref{lem:monotonicity+step}).
We will show that we can restrict our search for optimal cut lengths to these
values, and how to do away with many of them for efficiency. 

Combining the two ideas, we will consider candidate multisets of the following form.
\begin{definition}%
  \pdfbookmark[2]{Candidate Sets}{def:candidates}%
  \label{def:candidates}%
  We define the candidate multiset(s)
    \[ \candidatesmultiset(I,f_l, f_u) \wwrel\ce  
         \biguplus_{i \in I}\ 
         \biggl\{ \frac{L_i}{j} 
           \biggm| f_l(i) \leq j \leq f_u(i)
         \biggr\} \]
  dependent on index set~$I \subseteq [1..n]$ and functions
    $f_l : I \to \N$ and 
    $f_u : I \to \N \cup \{\infty\}$
  which bound the denominator from below and above, respectively; either
  may implicitly depend on $\mset L$ and/or $k$.
\end{definition}
Note that $|\candidatesmultiset| = \sum_{i \in I} [f_u(i) - f_l(i) + 1]$.
We denote the multiset of all candidates by
$\allcandidatesmultiset \ce \candidatesmultiset([1..n],1,\infty)$.

First, let us note that this definition covers the optimal solution as long as
upper and lower bounds are chosen appropriately.
\begin{lemma}\label{lem:optatjump}
  There is an optimal solution on a jump discontinuity of \Lpieces, i.\,e.\ 
  $\optlength \in \allcandidatesmultiset$.
\end{lemma}
\begin{proof}
  From its definition, we know that \Feasible has exactly one jump 
  discontinuity, and from \wref{lem:monotonicity+step} (via \Lpieces)
  we know that it is one of the $\nicefrac{L_i}{j}$. By \wref{cor:optismaxfeasible} 
  and left"=continuity of \Feasible (again via \Lpieces) we know that this is 
  indeed our solution~\optlength.
\end{proof}
Of course, our all"=encompassing candidate multiset $\allcandidatesmultiset$ is 
infinite (as is the corresponding set) and does hence not lend itself to a 
simple search. But there is hope:
we already know that $\optlength \geq l$ for any feasible
$l$ which immediately implies finite (albeit possibly large) 
bounds on $j$ (if we have such $l$).
We will now show how to restrict the set of candidates via suitable index sets
$I$ and bounding functions $f_l$ and $f_u$ so that we can efficiently search 
for~\optlength.
We have to be careful not to inadvertently remove \optlength by choosing bad
bounding functions.
\begin{lemma}
  \pdfbookmark[2]{Admissible Bounds}{lem:admissible-bounds}%
  \label{lem:admissible-bounds}%
  Let $I \subseteq [1..n]$ and $f_l,f_u : I \to \N$ so that 
  \begin{enumerate}[label=\roman*)]
    \item\label{item:bound.lowergood}%
      $f_l(i) = 1$ or 
      $\nicefrac{L_i}{(f_l(i) - 1)}$ is infeasible,  
      and
    \item\label{item:bound.uppergood}%
      $\nicefrac{L_i}{f_u(i)}$ is feasible,
  \end{enumerate}
  for all $i \in I$, and 
  \begin{enumerate}[label=\roman*),start=3]
    \item\label{item:bound.indicesgood}%
      $L_{i'}$ is suboptimal (i.\,e.\ $L_{i'}$ is feasible, but not optimal)
  \end{enumerate}
  for all $i' \in [1..n] \setminus I$.
  Then,
    \[ \optlength \in \candidatesmultiset(I, f_l, f_u) . \]
\end{lemma}
\begin{proof}
  We argue that $\allcandidatesmultiset \setminus \candidatesmultiset(I, f_l, f_u)$ 
  does \emph{not} contain the optimal solution~\optlength;
  the claim then follows with \wref{lem:optatjump}.
  
  Let $i \in [1..n]$ and $j \in [1..\infty]$ be arbitrary but fixed. We 
  investigate three cases for why length $\nicefrac{L_i}{j}$ may not be included
  in $\candidatesmultiset(I, f_l, f_u)$.
  \begin{description}[bmargin={\boldmath$j > f_u(i)$:}]
    \item[\boldmath$i \notin I$:]
      Candidate $\nicefrac{L_{i}}{j}$ is suboptimal by \wref{cor:optismaxfeasible} because
      $\nicefrac{L_{i}}{j} \leq L_{i}$ and $L_i$ itself is already suboptimal
      by~\ref{item:bound.indicesgood}.
      
    \item[\boldmath$j < f_l(i)$:]
      In this case, we must have $f_l(i)>1$, 
      so $\nicefrac{L_i}{(f_l(i)-1)}$ is infeasible by \ref{item:bound.lowergood}.
      Clearly, $\nicefrac{L_i}{j} > \nicefrac{L_i}{f_l(i)}$,
      so $\nicefrac{L_i}{j} \ge \nicefrac{L_i}{(f_l(i)-1)}$ and
      we get by monotonicity of \Feasible (cf.\ \wref{lem:monotonicity+step} via \Lpieces)
      that $\nicefrac{L_i}{j}$ is infeasible, as well.
      
    \item[\boldmath$j > f_u(i)$:] Clearly, $\nicefrac{L_i}{j} < \nicefrac{L_i}{f_u(i)}$,
      where the latter is already feasible by~\ref{item:bound.uppergood}.
      So, again by \wref{cor:optismaxfeasible}, $\nicefrac{L_i}j$ is suboptimal.
  \end{description}
  Thus, we have shown that every candidate length $\nicefrac{L_i}{j}$
  given by $(i,j) \in I \times [1..\infty]$ is either in 
  $\candidatesmultiset(I, f_l, f_u)$ or, failing that, infeasible or suboptimal.
\end{proof}
We will call triples $(I, f_l, f_u)$ of index set and bounding functions that 
fulfill \wref{lem:admissible-bounds} \emph{admissible restrictions} 
(for $\mset L$ and $k$). We say that $\candidatesmultiset(I, f_l, f_u)$ is
admissible if $(I, f_l, f_u)$ is an admissible restriction.

We will restrict ourselves for the remainder of this article to index sets~$\Ico$ 
that contain indices of lengths that are larger than the 
$n'$th~largest\footnote{%
  We borrow from the common notation $S_{(k)}$ for the $k$th~smallest element 
  of sequence $S$.
} length~$L^{(n')}$ in $\mset L$, for $n' = \min(k, n+1)$. 
This corresponds to working with $I_{> L^{(n')}}$ as defined in 
\wref{lem:feasible-computation}. We will have to show that such index sets are
indeed admissible (alongside suitable bounding functions); intuitively, if
$k \leq n$ then $L^{(k)}$ is always feasible, and otherwise we have to work with
all input lengths. We fix this convention for clarity and notational ease.
\begin{definition}%
  \pdfbookmark[2]{Cut-off length and index set}{def:ourindexset}%
  \label{def:ourindexset}%
  Define cut"=off length $\Lco$ by
    \[ 
       \Lco \wwrel\ce \begin{cases}
                L^{(k)}, &k \leq n; \\
                0,       &k > n, 
              \end{cases} 
    \]
  and index set $\Ico \subseteq [1..n]$ as
    \[ \Ico \wwrel\ce \begin{cases}
                 I_{> \Lco},        &\Lco \text{ not optimal}; \\
                 \text{undefined}, &\text{otherwise}. 
               \end{cases} \]
  Note that $\Ico = [1..n]$ if $k > n$.
\end{definition}
We will later see that we never invoke the undefined case as we already have 
$\optlength = L^{(k)}$ then.

In order to illustrate that we have found a useful criterion for admissible
bounds, let us investigate shortly an admittedly rather obvious choice of
bounding functions. We use the null"=bound $f_l \!\rel= i \mapsto 1$ and
$f_u \!\rel= i \mapsto k$; an optimal solution does not cut more than $k$ 
(equal"=sized) pieces out of any one stick. 
The restriction $([1..n],1,k)$ is clearly admissible; in particular, every 
$\nicefrac{L_i}{k}$ is feasible. 

\begin{examplectd}{ex:nonaivealg}
  \def\go#1{{\color{black!50}#1}}
  For $\mset{L}_{\mathrm{ex}}$ and $k=9$, we get
  \begin{align*}
    \candidatesmultiset(\Ico,1,k) = 
      \biggl\{ &\frac81,\frac82,\frac83,\frac84,\frac85,\frac86,\frac87,\frac88,\frac89,
             \;\frac71,\frac72,\frac73,\frac74,\frac75,\frac76,\go{\frac77},\frac78,\frac79,
            \\[.5ex]
              &\frac61,\frac62,\go{\frac63},\frac64,\frac65,\go{\frac66},\frac67,\frac68,\frac69,
             \;\go{\frac11},\frac12,\frac13,\frac14,\frac15,\frac16,\frac17,\frac18,\frac19 
              \biggr\} ,
  \end{align*}
  that is 36 candidates. Note that there are four duplicates, so there are
  32 distinct candidates.
\end{examplectd}

We give a full proof of admissibility and worst"=case size here; 
it is illustrative even if simple because later proofs will follow
the same structure.

\begin{lemma}%
  \pdfbookmark[2]{Quadratic Candidate Set}{lem:quadr-candidates}%
  \label{lem:quadr-candidates}%
  If $\Lco \neq \optlength$, then
  $\candidatesmultiset(\Ico,1,k)$ is admissible. \\
  Furthermore,
  $|\candidatesmultiset(\Ico,1,k)| = k \cdot \min(k-1,n) 
                                  \in \Th\bigl( k \cdot \min(k,n) \bigr)$
  in the worst case.
\end{lemma}
\begin{proof}
  First, we show that $(\Ico, 1, k)$ is an admissible restriction 
  (cf.\ \wref{lem:admissible-bounds}).
  \begin{description}[margin={ad \ref{item:bound.indicesgood}:}]
    \item[ad \ref{item:bound.lowergood}:]%
      Since $f_l(i) = 1$ for all $i$, this checks out.
    \item[ad \ref{item:bound.uppergood}:]%
      Clearly, $\Lpieces(\nicefrac{L_i}{k}) \geq k$ just from the contribution 
      of summand $\lfloor \nicefrac{L_i}{l} \rfloor$.
    \item[ad \ref{item:bound.indicesgood}:]%
      We distinguish the two cases of $\Lco$ (cf.\ \wref{def:ourindexset}).
      \begin{itemize}
        \item If $k > n$ then $\Ico = [1..n]$ which trivially
          fulfills \ref{item:bound.indicesgood}.
        \item In the other case, $k \leq n$ and $\Lco$ is not optimal by assumption.
          Then $\Ico = I_{> \Lco}$; therefore $L_{i'} \leq \Lco$
          for any $i' \notin \Ico$ and \wref{lem:monotonicity+step} implies that
          $L_{i'}$ is not optimal as well.
      \end{itemize}
  \end{description}
  This concludes the proof of the first claim.
  
  As for the number of candidates, note that clearly 
    $|\candidatesmultiset(\Ico, 1, k)| = \sum_{i \in \Ico} k = |\Ico| \cdot k$;
  the claim follows with $|\Ico| = \min(k-1, n)$ in the case that the $L_i$
  are pairwise distinct (cf.\ \wref{def:ourindexset}).
\end{proof}

Since we now know that we have to search only a finite domain for \optlength,
we can start thinking about effective and even efficient algorithms.

\section{Algorithms}
\label{sec:algorithms}

Just from the discussion above, a fairly elementary algorithm presents itself:
first cut the input down to the lengths given by $\Ico$ (cf.\ \wref{def:ourindexset}), 
then use binary search on the candidate set w.\,r.\,t.\ \Feasible. 
This works because \Feasible is non"=increasing (cf.\ \wref{cor:optismaxfeasible} 
and \wref{lem:monotonicity+step}).

\begin{algorithm}{$\alg{f_l, f_u}(\mset L, k) :$}
  \pdfbookmark[2]{Algorithm \algascii}{alg:searchalg}%
  \label{alg:searchalg}
  ~
  \begin{enumerate}
    \item\label{srcline:searchalg.rankbound}
      \gdef\algorithmsRankbound{%
        Compute $n' \ce \min(k, n+1)$.
      }\algorithmsRankbound
    \item\label{srcline:searchalg.downsize}
      \def\algorithmsLabelPrefix{searchalg}%
      \gdef\algorithmsDownsize{%
        If $n' \leq n$:
        \begin{enumerate}[label=2.\arabic*.,margin={2.2.}]
          \item\label{srcline:\algorithmsLabelPrefix.select}%
            Determine $\Lco \ce L^{(n')}$, i.\,e.\ the $n'$th~largest length.
          \item\label{srcline:\algorithmsLabelPrefix.optcheck}%
            If $\Lco$ is optimal, answer $\optlength = \Lco$ (and terminate).
        \end{enumerate}
      \item[2$'$\!.] Otherwise (i.\,e. $n' > n$):
        \begin{enumerate}[label=2.\arabic*.,start=3,margin={2.4.}]
          \item\label{srcline:\algorithmsLabelPrefix.simpleindices}%
            Set $\Lco \ce 0$.
        \end{enumerate}
      \item\label{srcline:\algorithmsLabelPrefix.partition}%
            Assemble $\Ico \ce I_{> \Lco}$.
      }\algorithmsDownsize
    \item\label{srcline:searchalg.candidates}%
      Compute $\candidatesmultiset \ce \candidatesmultiset(\Ico, f_l, f_u)$ as sorted array.
    \item\label{srcline:searchalg.search}%
      Find \optlength by binary search on $\candidatesmultiset$ w.\,r.\,t.\ 
      \Feasible.
    \item\label{srcline:searchalg.return}%
      Answer \optlength.
  \end{enumerate}
\end{algorithm}

For completeness we specify that $f_l, f_u : \Ico \to \N$.

\begin{theorem}\label{thm:algorithm}
  Let $(\Ico, f_l, f_u)$ be an admissible restriction where $f_l$ and $f_u$ can be 
  evaluated in time $\Oh(1)$.
  
  Then, algorithm $\alg{f_l, f_u}$ solves \optproblem in (worst"=case) time
    \[ T(n,k) 
        \in \Th(  n 
                + |\candidatesmultiset| \log |\candidatesmultiset| ) \]
  and space
    \[ S(n,k) \in \Th(n + |\candidatesmultiset|). \]
\end{theorem}

\begin{proof} We deal with the three claims separately.
  \begin{description}[margin={Correctness}]
    \item[Correctness] follows immediately from \wref{lem:admissible-bounds} and
      \wref{lem:monotonicity+step} resp.\ \wref{cor:optismaxfeasible}.
      Note in particular that $\alg{}$ does indeed compute $\Ico$ as defined in
      \wref{def:ourindexset}, and the undefined case is never reached.
  
    \item[Runtime:] Since the algorithm contains neither loops nor recursion
      (at the top level) we can analyze every step on itself.
      \begin{description}[margin={Steps \ref{srcline:searchalg.rankbound}, \ref{srcline:searchalg.simpleindices}:},itemsep=1ex]
        \item[Steps \ref{srcline:searchalg.rankbound},
                    \ref{srcline:searchalg.simpleindices}:]
          These clearly take time $\Oh(1)$.
          
        \item[Step \ref{srcline:searchalg.select}:]
           There are well"=known algorithms that perform selection
           in worst"=case time $\Th(n)$.
           
        \item[Step \ref{srcline:searchalg.optcheck}:]
          Testing $\Lco$ for optimality is as easy as computing $\Feasible(\Lco)$
          and counting the number $a$ of integral $\nicefrac{L_i}{\Lco}$ in 
          $\Lpieces(\Lco)$. If $\Feasible(\Lco)$ (i.\,e., $\Lpieces(\Lco) \geq k$) and 
          $\Lpieces(\Lco) - a < k$, then $\Lco$ is the jump discontinuity of \Feasible
          and $\Lco$ is optimal; otherwise it is not.
          
          Thus, this step takes time $\Th(n)$.
          
        \item[Step \ref{srcline:searchalg.partition}:]
          This can be implemented by a simple iteration over $[1..n]$ with
          a constant"=time length check per entry, hence in time $\Th(n)$.
          
          $\Ico$ can be assembled by one traversal over $\mset L$ and
          stored as simple linked list in (worst"=case) time $\Th(n)$.
        
        \item[Step \ref{srcline:searchalg.candidates}:] 
          By \wref{def:candidates} we have $|\candidatesmultiset|$ many 
          candidates. Sorting these takes time 
          $\Th(|\candidatesmultiset| \log |\candidatesmultiset|)$
          using e.\,g.\ Heapsort.
        
        \item[Step \ref{srcline:searchalg.search}:]
          The binary search clearly takes at most 
            $\lfloor \log_2 |\candidatesmultiset|+1 \rfloor$ 
          steps.
          In each step, we evaluate \Feasible in time 
          \begin{itemize}
            \item $\Th(|\Ico|)$ for all candidates $l > \Lco$ using
              \wref{lem:feasible-computation}, and
            \item $\Oh(1)$ for $l \leq \Lco$ since we already know from 
              feasibility of $\Lco$ via \wref{lem:monotonicity+step}
              that these are feasible, too.
          \end{itemize}
          Therefore, this step needs time $\Th(|\Ico| \log |\candidatesmultiset|)$ 
          time in total. 
          
          It is easy to see that admissible bounds always fulfill $f_u(i) \geq f_l(i)$ 
          for all $i \in \Ico$. Therefore, $|\Ico| \leq |\candidatesmultiset|$ so
          the runtime of this step is dominated by step~\ref{srcline:searchalg.candidates}.          
      \end{description}

    \item[Space:] The algorithm stores $\mset L$ of size $\Th(n)$, plus maybe a 
      copy for selection and partitioning (depends the actual algorithm used). 
      Step~\ref{srcline:searchalg.candidates}
      then creates a $\Th(|\candidatesmultiset|)$"/large representation of the 
      candidate set. 
      Both step~\ref{srcline:searchalg.partition} and~\ref{srcline:searchalg.search}
      can be implemented iteratively, and a potential recursion depth 
      (and therefore stack size) in step~\ref{srcline:searchalg.select} 
      is bounded from above by its runtime $\Oh(n)$.
      A few additional auxiliary variables require only constant amount of memory.
  \end{description}
\end{proof}

For practical purposes, eliminating duplicates in Step~\ref{srcline:searchalg.candidates}
is virtually free and can speed up the subsequent search. In the worst case,
however, we save at most a constant factor with the bounding functions we
consider (see \wref{app:additional-results}), so we decided to stick to the 
clearer presentation using multisets (instead of candidate \emph{sets}).

\subsection{Knowing Beats Searching}
\label{sec:selection-algorithm}

We have seen that the runtime of algorithm \alg{} is dominated by \emph{sorting}
the candidate set. This is necessary for facilitating binary search; but do we 
\emph{have} to search? 
As it turns out, a slightly different point of view on the problem allows us 
to work with the unsorted candidate multiset and we can save a 
factor~$\log |\candidatesmultiset|$.

The main observation is that \Lpieces increases its value by one at
every jump discontinuity (for each $\nicefrac{L_i}{j}$ that has that same value).
So, knowing $\Lpieces(l)$ for any candidate length $l$, we know exactly how many candidates
(counting duplicates) we have to move to get to the jump of \Feasible. 
Therefore, we can make do with \emph{selecting} the solution from our candidate 
set instead of searching through it.

The following lemma states the simple observation that $\lpieces(L,l)$ is 
intimately related to the ``position'' of $l$ in the decreasingly sorted 
candidate multiset for $L$.

\begin{lemma} 
  For all $L,l \in \Q_{>0}$,
   \[ \lpieces(L,l) \rel= 
		    \bigl|\bigl\{ \nicefrac{L}{j} \mid j \in \N \land \nicefrac{L}{j} \ge l \bigr\}\bigr| . \]
\end{lemma}
\begin{proof}
	The right"=hand side equals the largest integer $j \in \N_0$ for which
	$\nicefrac{L}{j} \ge l$, i.\,e.\ $j \le \nicefrac{L}{l}$, which is
	by definition $\lfloor \nicefrac{L}{l} \rfloor = \Lpieces(L,l)$. 
	Note that this argument extends to the case $L < l$ by formally 
	setting $\nicefrac{L}{0} = \infty \ge l$.
\end{proof}

Since we consider multisets, we can lift this property to $\Lpieces(l)$:

\begin{corollary}\label{cor:lpieces-counts-jumps}
  For all $l \in \Q_{>0}$,
	  \[ \Lpieces(l) \wwrel=
		     \sum_{i=1}^n 
			     \bigl|\bigl\{ \nicefrac {L_i}j \mid 
				     j\in\N \land \nicefrac {L_i}j \ge l \bigr\}\bigr|
		     \wwrel=
			     \bigl| \allcandidatesmultiset \cap [l,\infty) \bigr|
	  .\]
\qed\end{corollary}

In other words, $\Lpieces(l)$ is the number of occurrences of candidates 
that are at least $l$.
We can use this to transform our search problem (cf.\ \wref{cor:optismaxfeasible})
into a selection problem.
\begin{lemma}\label{lem:opt-is-kth-cand}
  $\optlength = \allcandidatesmultiset^{(k)}$.
\end{lemma}
\begin{proof}
  Denote with $\allcandidates$ the \emph{set} of all candidates, that is
  $l \in \allcandidates \iff \allcandidatesmultiset(l) > 0$.
  We can thus write the statement of \wref{cor:lpieces-counts-jumps} as
  \begin{equation}\label{eq:lpieces-by-occurs}
    \Lpieces(l) \wwrel= \sum_{\substack{l' \,\in\, \allcandidates\\l' \,\geq\, l}} 
                    \allcandidatesmultiset(l') .
  \end{equation}
  As a direct consequence, we get for every $i \in \N$ that
  \begin{equation}\label{eq:lpieces-bound-by-rank+occur}
      i 
    \wwrel\leq 
      \Lpieces\bigl( \allcandidatesmultiset^{(i)} \bigr) 
    \wwrel\leq 
      i + \allcandidatesmultiset\bigl( \allcandidatesmultiset^{(i)} \bigr) 
        - 1;
  \end{equation}
  see \wref{fig:selection-sketch} for a sketch of the situation.
  \begin{figure}
    \begin{center}
      \begin{tikzpicture}[x=10mm,y=-7mm]
        \foreach \x/\lab/\i/\m in {%
             0/{$i$}/{$l_i$}/{$\Lpieces(l_i)$},
             1/{$1$}/{$10$}/{$1$},%
             2/{$2$}/{$8$}/{$4$},%
             3/{$3$}/{$8$}/{$4$},%
             4/{$4$}/{$8$}/{$4$},%
             5/{$5$}/{$5$}/{$5$},%
             6/{$6$}/{$4$}/{$7$},%
             7/{$7$}/{$4$}/{$7$}%
        } {%
          \ifthenelse{ \x > 0 }{%
            \draw[fill=black!10,thin] 
              ($(\x,0.5) + (-5mm,0)$) rectangle +(10mm,1); 
          }{}
          \node at (\x, 0) {\textsmaller{\lab}};
          \node at (\x, 1) {\i};
          \node at (\x, 2) {\m};
        }
        
        \fill[black!10] ($(8,0.5) + (-5mm,0)$) rectangle +(10mm,1);
        \draw[thin]
          ($(8,0.5) + (-5mm,0)$) -- +(10mm,0)
          ++(0,1) -- ++(10mm,0);
        \node at (8,1) {$\dots$};
      \end{tikzpicture}
    \end{center}
    \caption{An example illustrating \wref[eq.]{eq:lpieces-bound-by-rank+occur} 
      with $l_i = \allcandidatesmultiset^{(i)}$ for some suitable instance.
      Note that the lower bound is tight for $i \in \{1,5\}$ and the upper
      for $i = 2$.}
    \label{fig:selection-sketch}
  \end{figure}
  Feasibility of $l \ce \allcandidatesmultiset^{(k)}$ follows immediately. Now let
  $\hat{l} \ce \min \{ l' \in \allcandidatesmultiset \mid l' > l \}$;
  we see that
   \[ \Lpieces(\hat{l})
        \overset{\eqref{eq:lpieces-by-occurs}}{\wwrel=}
      \Lpieces(l) - \allcandidatesmultiset(l)
        \overset{\eqref{eq:lpieces-bound-by-rank+occur}}{\wwrel\leq}
      k + \allcandidatesmultiset(l) - 1 - \allcandidatesmultiset(l)
        \wwrel=
      k - 1 \]
  and therefore $\hat{l}$ is infeasible.
  By the choice of $\hat{l}$ and monotonicity of \Feasible (cf.\ \wref{lem:monotonicity+step})
  we get that $l = \allcandidatesmultiset^{(k)}$ is indeed the largest feasible candidate;
  this concludes the proof via \wref{cor:optismaxfeasible} and \wref{lem:optatjump}.
\end{proof}

Of course, we want to select from a small candidate set such as those we saw above;
surely, selecting the $k$th~largest element from these is not correct, in general.
Also, not all restrictions may allow us to select because if we miss an 
$\nicefrac{L_i}{j}$ between two others, we may count wrong.
The relation carries over to \emph{admissible} restrictions with only small 
adaptions, though.

\begin{corollary}\label{cor:opt-is-k'th-cand}
  Let $\candidatesmultiset = \candidatesmultiset(I, f_l, f_u)$ 
 	be an admissible candidate multiset.
 	Then,
    \[ \optlength \wrel= \candidatesmultiset^{(k')} \]
  with $k' = k - \sum_{i \in I} \bigl[ f_l(i) - 1 \bigr]$.
\end{corollary}
\begin{proof} 
  With multiset 
    \[ \mset{M} 
         \wwrel\ce 
       \biguplus_{i \,\in\, I}\ \{ \nicefrac{L_i}{j} \mid i \in I, j < f_l(i) \} , \] 
  we get by \wref{lem:admissible-bounds}~\ref{item:bound.uppergood} and 
  \wref{lem:monotonicity+step} that
    \[ \candidatesmultiset \cap [\optlength, \infty) 
         \wwrel=
       \bigl( \allcandidatesmultiset \cap [\optlength, \infty) \bigr)
         \setminus \mset{M}. \]
  In addition, we know from \wref{lem:opt-is-kth-cand} that
    \[ \bigl( \allcandidatesmultiset \cap [\optlength, \infty) \bigr)^{(k)}
         \wwrel=
       \optlength. \]
  Since $\mset{M}$ contains only infeasible candidates (cf.\ 
  \wref{lem:admissible-bounds}~\ref{item:bound.lowergood} and 
  \wref{lem:monotonicity+step}), we also have that
    \[ \mset{M} \subset (\optlength, \infty), \]
  and by definition
    \[ \mset{M} \cap \candidatesmultiset = \emptyset. \]
  The claim
    \[ \optlength \wwrel= \allcandidatesmultiset^{(k)}
                  \wwrel= \candidatesmultiset^{(k - |\mset{M}|)} \]
  follows by counting.
\end{proof}
Hence, we can use any of the candidate sets we have investigated above.
Instead of binary search we determine \optlength by selecting the $k'$th~largest 
element according to \wref{cor:opt-is-k'th-cand}.
Since selection takes only linear time we save a logarithmic factor compared 
to \alg{}.

We give the full algorithm for completeness; note that steps 
\ref{srcline:selalg.rankbound} and \ref{srcline:selalg.downsize} have not
changed compared to \alg{}.

\begin{algorithm}{$\coolalg{f_l, f_u}(\mset L, k) :$}
  \pdfbookmark[2]{Algorithm \coolalgascii}{thm:linear-alg}%
  \label{alg:selalg}
  ~
  \begin{enumerate}
    \item\label{srcline:selalg.rankbound}
      \algorithmsRankbound
    
    \item\label{srcline:selalg.downsize}
      \def\algorithmsLabelPrefix{selalg}%
      \algorithmsDownsize
      
    \item\label{srcline:selalg.candidates}
      Compute $\candidatesmultiset \ce \candidatesmultiset(\Ico, f_l, f_u)$
        as multiset.
        
    \item\label{srcline:selalg.search}
      Determine $k' \ce k - \sum_{i \in \Ico}\bigl[ f_l(i) - 1 \bigr]$.
      
    \item\label{srcline:selalg.return}
      Answer $\optlength \ce \candidatesmultiset^{(k')}$.
  \end{enumerate}
\end{algorithm}

\begin{theorem}\label{thm:linear-alg}%
  Let $(\Ico, f_l, f_u)$ be an admissible restriction where $f_l$ and $f_u$ can be 
  evaluated in time $\Oh(1)$.
  
  Then, \coolalg{f_l,f_u} solves \optproblem in time and space 
  $\Th(n + |\candidatesmultiset|)$.
\end{theorem}
\begin{proof}
  Correctness is clear from \wref{lem:smartbounds-admissible} and 
  \wref{cor:opt-is-k'th-cand}.
    
  We borrow from the resource analysis of \wref{thm:algorithm} with the 
  following changes.
  \begin{description}[margin={ad \ref{srcline:selalg.search},\ref{srcline:selalg.return}:}]
    \item[ad \ref{srcline:selalg.candidates}:]
      We do not sort $\candidatesmultiset$, so creating the multiset takes
      only time $\Th(|\candidatesmultiset|)$; 
      the result takes up space $\Th(|\candidatesmultiset|)$, too, though.
        
    \item[ad \ref{srcline:selalg.search},\ref{srcline:selalg.return}:]
      Instead of binary search on \candidatesmultiset with repeated evaluation of 
      \Feasible, we just have to compute $k'$ (which clearly takes time $\Th(|\Ico|)$)
      and then select the $k'$th~largest element from \candidatesmultiset.
      This takes time $\Th(|\candidatesmultiset|)$ using e.\,g.\ the 
      median"=of"=medians algorithm~\cite{Blum1973}.
  \end{description}
  The resource requirements of the other steps remain unchanged, that is $\Th(n)$. 
  The bounds we claim in the corollary follow directly.
\end{proof}
Is has become clear now that decreasing the number of candidates is crucial
for solving \optproblem quickly, provided we do not drop \optlength along the way. 
We now endeavor to do so by choosing better admissible bounding functions.

\section{Reducing the Number of Candidates}
\label{sec:candidate-sets}

We can decrease the number of candidates significantly by observing the following.
Whenever we cut $L^{(i)}$ (which is the $i$th~largest length) into $j$ pieces 
of length $\nicefrac{L^{(i)}}{j}$ each, 
we also get at least $j$ pieces of the same length from each of the longer sticks.
In total, this makes for at least $i \cdot j$ pieces of length $\nicefrac{L^{(i)}}{j}$;
see \wref{fig:visual-k/i} for a visualization.
By rearranging the inequality $k \geq i \cdot j$, we obtain a new admissible
bound on $j$. For the algorithm, we have to sort $\Ico$, though, so that $L_i = L^{(i)}$.

\begin{examplectd}{ex:nonaivealg}
  For $\mset{L}_{\mathrm{ex}}$ and $k=9$, we get
  \def\go#1{{\color{black!50}#1}}
  \begin{align*}
    \candidatesmultiset(\Ico,1,\lceil \nicefrac{k}{i} \rceil) = 
        \biggl\{ &\frac81,\frac82,\frac83,\frac84,\frac85,\frac86,\frac87,\frac88,\frac89,
               \;\frac71,\frac72,\frac73,\frac74,\frac75,
               \;\frac61,\frac62,\go{\frac63}
        \biggr\},
  \end{align*}
  that is 17 candidates (16 distinct ones); compare to $|\candidatesmultiset(\Ico,1,k)| = 36$). 
\end{examplectd}

\begin{figure}
  \begin{center}
    \begin{tikzpicture}[x=7mm,y=5mm,auto,scale=1.2,transform shape]
      \fill[black!25]
        \foreach \y in {1,2,5,6} {%
          \foreach \x in {1,2.5,5,6.5} {%
             (\x,\y) rectangle +(1,1)
          }
        };
      
      \draw[thin,black!25,fill=black!10]
        \foreach \x/\y in {1/7, 8/1, 8/2, 8/5} {%
          (\x, \y) rectangle +(1,1)
        };
        
      \draw[dashed,black!85]
        \foreach \y in {1,2,4,5,6} {%
           (1,\y+1) -- (7.5,\y+1)
        };
        
      \foreach \i/\x/\h in {1/1/7.5, 2/2.5/6.8, {i-1}/5/6.5, i/6.5/6, {i+1}/8/5.5} {%
        \draw (\x,1) rectangle +(1,\h);
        \node[scale=0.85] at (\x + 0.5,\h + 1.5) {$L^{(\i)}$};
      }
      
      \foreach \x in {4.25,9.65} {%
        \node at (\x,4) {$\,\dots$};
      }      
      \foreach \x in {1.5, 3, 5.5, 7} {%
        \node[scale=0.85,black!85] at (\x, 4.2) {$\vdots$};
      }
      \node[scale=0.85,black!25] at (8.5, 4.2) {$\vdots$};
      
      \draw[decoration={brace},decorate]
        (0.75,1) -- +(0,6) 
          node[midway,scale=0.85,outer xsep=1mm] {$\frac{k}{i}$};
      \draw[decoration={brace,mirror},decorate]
        (1,0.75) -- +(6.5,0) 
          node[midway,swap,scale=0.85,outer ysep=1mm] {$i$};
    \end{tikzpicture}
  \end{center}
  \caption{When considering cut lengths $\nicefrac{L^{(i)}}{j}$, no $j$ larger
    than $\lceil \nicefrac{k}{i} \rceil$ is relevant. 
    The sketch shows a cutting with $L^{(i)}$ and $j = \nicefrac{k}{i}$.
    Note how we have $k$ maximal pieces for sure (dark); 
    there may be many more (light).}
  \label{fig:visual-k/i}
\end{figure}
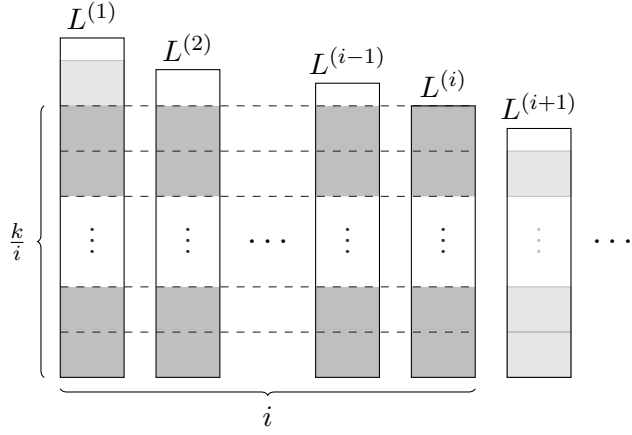

\begin{lemma}%
  \pdfbookmark[2]{Linearithmic Candidate Set}{lem:linearithmic-candidates}%
  \label{lem:linearithmic-candidates}%
  Assume that $\Lco \neq \optlength$ and $\Ico$ is sorted w.\,r.\,t.\ decreasing lengths.
  
  Then, $\candidatesmultiset(\Ico,1,\lceil \nicefrac{k}{i} \rceil)$ is admissible. \\
  Furthermore,
  $|\candidatesmultiset(\Ico,1,\lceil \nicefrac{k}{i} \rceil)| 
    \in \Th\bigl( k \cdot \log(\min(k,n)) \bigr)$ 
  in the worst case.
  
  %
\end{lemma}
\begin{proof}
  Again, we start by showing that $(\Ico, 1, \lceil \nicefrac{k}{i} \rceil)$ 
  is an admissible restriction.
  \begin{description}[margin={ad \ref{item:bound.lowergood}, \ref{item:bound.indicesgood}:},itemsep=1ex]
    \item[ad \ref{item:bound.lowergood}, \ref{item:bound.indicesgood}:]%
      Similar to the proof of \wref{lem:quadr-candidates}.
      
    \item[ad \ref{item:bound.uppergood}:]%
      Because $\Ico$ is sorted, we have $L_i = L^{(i)}$ and $L_{i'} \geq L_i$ for
      $i' \leq i$. Therefore, we get for all the 
      $l = \nicefrac{L_i}{f_u(i)} = L_i \cdot \lceil \nicefrac{k}{i} \rceil^{-1}$
      with $i \in \Ico$ that
        \[            \Lpieces(L) 
           \wwrel=    \sum_{i' = 1}^n \Biggl\lfloor \frac{L_{i'}}{l} \Biggr\rfloor 
           \wwrel\geq \sum_{i' = 1}^i \Biggl\lfloor \frac{L_{i}}{l} \Biggr\rfloor 
           \wwrel=    \sum_{i' = 1}^i \Biggl\lfloor \biggl\lceil 
                                        \frac{k}{i} 
                                      \biggr\rceil \Biggr\rfloor 
           \wwrel\geq k. \]
  \end{description}
  This concludes the proof of the first claim.
  
  For the size bound, let for short 
  $\candidatesmultiset \ce \candidatesmultiset(\Ico,1,\lceil \nicefrac{k}{i} \rceil)$.
  Clearly, $|\candidatesmultiset| = \sum_{i \in \Ico} \lceil \nicefrac{k}{i} \rceil$
  (cf.\ \wref{def:candidates}).  
  With $|\Ico| = n' - 1 = \min(n,k-1)$ in the worst"=case 
  (cf.\ the proof of \wref{lem:quadr-candidates}), 
  the $\Th(k \log n')$ bound on $|\candidatesmultiset|$ follows from
    \[ |\candidatesmultiset|
         \wwrel=    \sum_{i=1}^{n'-1} \biggl\lceil \frac{k}{i} \biggr\rceil
         \wwrel\leq n' + \sum_{i=1}^{n'} \frac{k}{i}
         \wwrel=    n' + k \cdot \harm{n'} 
         \wwrel\in  \Th(k \log n') \]
  and
    \[ |\candidatesmultiset|
         \wwrel=    \sum_{i=1}^{n'-1} \biggl\lceil \frac{k}{i} \biggr\rceil
         \wwrel\geq \sum_{i=1}^{n'-1} \frac{k}{i}
         \wwrel=    k \cdot \harm{n'-1}  
         \wwrel\in  \Th(k \log n') \]  
  with the well"=known asymptotic $\harm{k} \sim \ln k$ 
  of the harmonic numbers~\cite[eq.~(6.66)]{ConcreteMathematics}.
\end{proof}

Combining \wref{thm:linear-alg} and \wref{lem:linearithmic-candidates} we have 
obtained an algorithm that takes time and space $\Th\bigl(n + k \cdot \log(\min(k,n))\bigr)$.
This is already quite efficient.
By putting in some more work, however, we can save the last logarithmic factor
that separates us from linear time and space.

\subsection{Sandwich Bounds}
\label{sec:linear-candidate-set}

\wref{lem:admissible-bounds} gives us some idea about what criteria we can use
for restricting the set of lengths we investigate. We will now try to match these
criteria as exactly as possible, deriving an interval 
$[\underline{l}, \overline{l}] \subseteq \Q_{>0}$ that includes \optlength and
is as small as possible; from these, we can infer almost as tight bounds $(f_l, f_u)$.

Assume we have some length $L < \optlength$ and consider only lengths $l > L$.
We have seen in \wref{lem:feasible-computation} that we can then restrict ourselves
to lengths from $I_{>L}$ when computing $\Lpieces(l)$. Now, from the definition 
of \Lpieces it is clear that we can sandwich $\Lpieces(l)$ by 
  \[            \sum_{i \,\in\, I_{>L}} \frac{L_i}{l} - 1
     \wwrel<    \Lpieces(l) 
     \wwrel\leq \sum_{i \,\in\, I_{>L}} \frac{L_i}{l}  \]
for $l > L$. We denote for short $\Sigma_{I} \ce \sum_{i \in I} L_i$ for any 
$I \subseteq [1..n]$; rearranging terms, we can thus express these bounds more easily,
both with respect to notational and computational effort. We get
\begin{equation}\label{eq:convenient-lpieces-bounds}%
             \frac{\Sigma_{I_{>L}}}{l} - |I_{>L}| 
  \wwrel<    \Lpieces(l)
  \wwrel\leq \frac{\Sigma_{I_{>L}}}{l}
\end{equation}
for all $l > L$. Note that $L=0$ is a valid choice, as then simply $I_{>L} = [1..n]$.

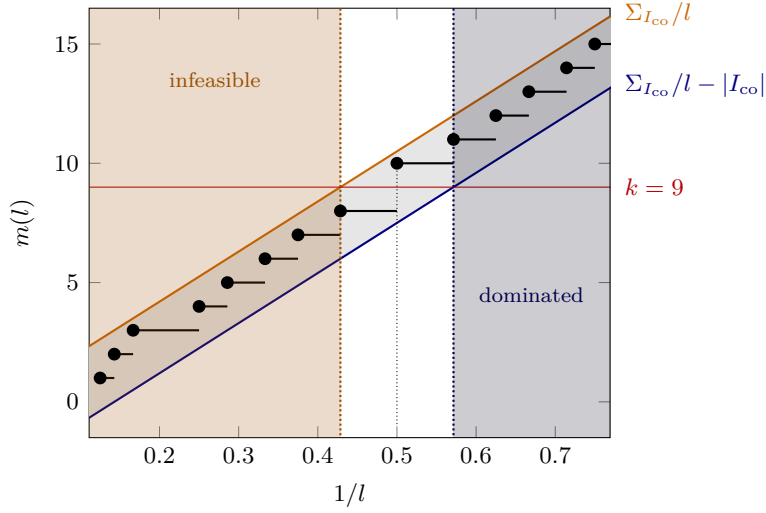
\begin{figure}
	\pgfplotstableread{
	l	1/l	m(l)	c(l)	
	1.	1.	34.	18.	
	1.14286	0.875	18.	17.	
	1.16667	0.857143	17.	16.	
	1.2	0.833333	16.	15.	
	1.33333	0.75	15.	14.	
	1.4	0.714286	14.	13.	
	1.5	0.666667	13.	12.	
	1.6	0.625	12.	11.	
	1.75	0.571429	11.	10.	
	2.	0.5	10.	8.	
	2.33333	0.428571	8.	7.	
	2.66667	0.375	7.	6.	
	3.	0.333333	6.	5.	
	3.5	0.285714	5.	4.	
	4.	0.25	4.	3.	
	6.	0.166667	3.	2.	
	7.	0.142857	2.	1.	
	8.	0.125	1.	0.	
}\plotpointsLpiecesMfourXtwo
\plaincenter{%
\begin{tikzpicture}[
	  remember picture,
	  every node/.style={font={\footnotesize}}
  ]
  \begin{axis}[%
	  xmax={1/1.3},
	  xmin={1/9},
	  xlabel={$1/l$},
	  ylabel=$\Lpieces(l)$,
  ]
	  \begin{scope}
	  \clip (axis cs:.112,-1) rectangle (axis cs:.768,17);
	  \fill[black!15] (axis cs:0,0) -- (axis cs:1,21*1) -- (axis cs:1,21-3) -- (axis cs:0,-3) -- cycle;
	  \end{scope}
	  \draw[thick,orange!80!black] (axis cs:0,0) -- (axis cs:1,21) 
	    coordinate[pos=0.775] (legend-upper);
	  \draw[thick,blue!50!black] (axis cs:0,-3) -- (axis cs:1,21*1-3) 
	    coordinate[pos=0.775] (legend-lower);

	  \addplot [
	      thick,
	      jump mark right,
	      mark=*,
	  ] table [x=1/l,y=m(l)] {\plotpointsLpiecesMfourXtwo};
	  \addplot[only marks,mark=*,mark options={fill=white},opacity=0] 
		   table [x=1/l,y=c(l)] {\plotpointsLpiecesMfourXtwo};
	
	  \draw[red!70!black] (axis cs:0,9) -- (axis cs:1,9)
	    coordinate[pos=0.775] (legend-k);

	  \draw[thin, densely dotted] (axis cs:1/2,10) -- (axis cs:1/2,-1) ;
	
	  \fill[opacity=0.2,orange!60!black] (axis cs:0,20) rectangle (axis cs:3/7,-10) ;
	  \draw[orange!60!black,thick,densely dotted] (axis cs:3/7,-3) -- (axis cs:3/7,20) 
		  	node[pos=.175,anchor=east,inner sep=1pt] {$\nicefrac 1{\overline l}$}
		  (0.27,13.5) node {\scriptsize infeasible};

	  \fill[opacity=0.2,blue!50!black] (axis cs:1,20) rectangle (axis cs:4/7,-3) ;
	  \draw[blue!30!black,thick,densely dotted] (axis cs:4/7,-3) -- (axis cs:4/7,20) 
	  		node[pos=.175,anchor=west,inner sep=1pt] {$\nicefrac 1{\underline l}$}
		  (0.67,4.5) node {\scriptsize dominated};
  \end{axis}
  \node[orange!80!black,anchor=west] at (legend-upper) {$\Sigma_{\Ico}/l$};
  \node[blue!50!black,anchor=west] at (legend-lower) {$\Sigma_{\Ico}/l-|\Ico|$};
  \node[red!70!black,anchor=west] at (legend-k) {$k=9$};
\end{tikzpicture}%
}
	\caption{%
		The number of maximal pieces $\Lpieces(l)$ in the reciprocal of cut length $l$
		for $(\mset{L}_{\mathrm{ex}}, 9)$ as defined in \wref{ex:nonaivealg}.
		Note how we can exclude all but three candidates (the filled circles) 
		in a narrow corridor around $\nicefrac{1}{\optlength} = 0.5$, 
		defined by the points at which the bounds from \wref{eq:convenient-lpieces-bounds}
		attain $k=9$, namely $\underline l = 1.75$ and $\overline l = 2.\overline 3$.
	}
	\label{fig:corridor-bounds-example}
\end{figure}

Rearranging these inequalities ``around'' $\Lpieces(l) = k$
yields bounds on \optlength, which we can translate into bounds $(f_l, f_u)$ on 
$j$ (cf.\ \wref{def:candidates}). We lose some precision because we round to integer
bounds but that adds at most a linear number of candidates.
A small technical hurdle is to ensure that both bounds are greater than our chosen
$L$ so that we can apply the sandwich bounds \eqref{eq:convenient-lpieces-bounds} 
in our proof.
\begin{lemma}\label{lem:smartbounds-admissible}%
  Let $\Lco$ and $\Ico$ be defined as in \wref{def:ourindexset}, and
  \begin{itemize}
    \item $\displaystyle 
      \underline{l} \ce \max \biggl\{ \Lco, \frac{\Sigma_{\Ico}}{k + |\Ico|} \biggr\}$ 
      and
    \item $\displaystyle \overline{l}  \ce \frac{\Sigma_{\Ico}}{k}$.
  \end{itemize}
  Then, $\bigl( \Ico, \allpieces(L_i,\overline{l}), 
                     \allpieces(L_i,\underline{l}) \bigr)$ 
  is admissible.
\end{lemma}
\begin{proof}
  First, we determine what we know about our length bounds.
  Recall that $\Ico = I_{> \Lco} \neq \emptyset$ and $\Lco$ is not optimal.
  
  We see that $\underline{l}$ is feasible by calculating
  \begin{equation}\label{eq:lower-l-feasible}%
      \Lpieces(\underline{l})\ 
      \begin{cases} 
        \wwrel{\overset{\eqref{eq:convenient-lpieces-bounds}}{>}}
          \frac{\Sigma_{\Ico}}{\underline{l}} - |\Ico|
        \wwrel=
          \frac{\Sigma_{\Ico}}{\frac{\Sigma_{\Ico}}{k + |\Ico|}} - |\Ico|
        \wwrel=
          k, &\underline{l} > \Lco, \\
        
        \wwrel=    \Lpieces(\Lco) 
        \wwrel\geq k, &\underline{l} = \Lco > 0,
      \end{cases}
  \end{equation}
  using in the second case that $\Lco$ is feasible. For the upper bound, we
  first note that because $\Lco$ is not optimal, there is some $\delta > 0$ with
    \[ \Sigma_{\Ico} \wwrel\geq k(\Lco + \delta) \wwrel> k\Lco, \] 
  from which we get by rearranging that $\overline{l} > \Lco$.
  Therefore, we can bound
  \begin{equation}\label{eq:upper-l-feasibility-border}%
        \Lpieces(\overline{l} + \varepsilon)
    \wwrel{\overset{\eqref{eq:convenient-lpieces-bounds}}{\leq}}
        \frac{\Sigma_{\Ico}}{\overline{l} + \varepsilon}
    \wwrel<
        \frac{\Sigma_{\Ico}}{\overline{l}}
    \wwrel=
        \frac{\Sigma_{\Ico}}{\frac{\Sigma_{\Ico}}{k}}
    \wwrel=
        k
  \end{equation}
  for any $\varepsilon > 0$, that is any length larger than $\overline{l}$ is infeasible.
  Note in particular that, in every case, $\overline{l} > \underline{l}$ so we always
  have a non"=empty interval to work with.

  We now show the conditions of \wref{lem:admissible-bounds} one by one.
  \begin{description}[margin={ad \ref{item:bound.indicesgood}},itemsep=1ex]
    \item[ad \ref{item:bound.lowergood}]
      Let $i \in \Ico$.
      If $\allpieces(L_i, \overline{l}) = 1$ the condition is trivially fulfilled.
      In the other case, we calculate
        \[ l \wwrel\ce \frac{L_i}{\allpieces(L_i,\overline{l}) - 1}
             \wwrel=   \frac{L_i}{\lceil \nicefrac{L_i}{\overline{l}} \rceil - 1}
             \wwrel>   \frac{L_i}{\nicefrac{L_i}{\overline{l}}} 
             \wwrel=   \overline{l} \]
        and therewith $\Lpieces(l) < k$ by \eqref{eq:upper-l-feasibility-border}.
        
    \item[ad \ref{item:bound.uppergood}]
      Let $i \in \Ico$ again. We calculate
        \[ l \wwrel\ce  \frac{L_i}{\allpieces(L_i, \underline{l})}
             \wwrel=    \frac{L_i}{\lceil \nicefrac{L_i}{\underline{l}} \rceil}
             \wwrel\leq \frac{L_i}{\nicefrac{L_i}{\underline{l}}}
             \wwrel=    \underline{l} \]
      which implies by \wref{lem:monotonicity+step} that
        \[ \Lpieces(l) \wwrel\geq \Lpieces(\underline{l})
                       \wwrel{\overset{\eqref{eq:lower-l-feasible}}{\geq}} k .\]
                       
    \item[ad \ref{item:bound.indicesgood}] 
      See the proof of \wref{lem:quadr-candidates}.
  \end{description}
\end{proof}

\begin{examplectd}{ex:nonaivealg}
  For $\mset{L}_{\mathrm{ex}}$ and $k=9$, we get
  \def\go#1{{\color{black!50}#1}}%
   \[ \candidatesmultiset\bigl(\Ico,\allpieces(L_i,\overline{l}), 
                                   \allpieces(L_i,\underline{l})\bigr) = 
        \biggl\{ \frac{8}{4}, \frac{8}{5},\;
                \frac{7}{4}, \frac{7}{3},\;
                \go{\frac{6}{3}}, \frac{6}{4} \biggr\} , \]      
  that is six candidates (five distinct ones);
  compare to $|\candidatesmultiset(\Ico,1,\lceil \nicefrac{k}{i} \rceil)| = 17$ and 
  $|\candidatesmultiset(\Ico,1,k)| = 36$. 
  See \wref{fig:corridor-bounds-example} for a visualization of the effect
  our bounds have on the candidate set; note that we keep some additional
  candidates smaller than $\underline{l}$.
\end{examplectd}

We see in this example that the bounds from \wref{lem:smartbounds-admissible} 
are not as tight as could be; $\candidatesmultiset(\Ico, \allpieces(L_i,\overline{l}), 
\allpieces(L_i,\underline{l})) \cap [\underline{l}, \overline{l}]$ can be
properly smaller (but not by more than one element per $L_i$), 
and since $\optlength \in [\underline{l}, \overline{l}]$ it is still a valid 
candidate set.

We stick with the slightly larger set here for conciseness of the proofs,
but remark that omitting lengths outside the interval $[\underline l,\overline l]$
is safe.
We have defined admissibility in a way that is \emph{local} to each $L_i$~-- 
we require to envelop \optlength for each length in isolation: 
in particular, condition~\ref{item:bound.uppergood} ensures we include 
at least one $j$ for every $L_i$, so that $\nicefrac{L_i}j$ is feasible.
We thus have no way to express \emph{global} length bounds $[\underline l,\overline l]$
in this framework: although lengths smaller than $\underline l$ are dominated,
the upper bound $f_u \!\rel= i \mapsto \lfloor \nicefrac{L_i}{\underline{l}} \rfloor$ 
is not admissible in the sense of \wref{lem:admissible-bounds} 
because it might for some sticks not add a single feasible length.

Nevertheless, we have obtained yet another admissible restriction and, as it turns
out, it is good enough to achieve a linear candidate set. Only some 
combinatorics stand between us and our next corollary.

\begin{lemma}%
  \pdfbookmark[2]{Linear Candidate Set}{lem:linear-candidates}%
  \label{lem:linear-candidates}%
  $|\candidatesmultiset(\Ico,\allpieces(L_i,\overline{l}), \allpieces(L_i,\underline{l}))| 
    \in \Th\bigl( \min(k,n) \bigr)$ 
  in the worst case.
\end{lemma}
\begin{proof}
  Recall that $|\Ico| = \min(k-1,n)$ in the worst case 
  (cf.\ the proof of \wref{lem:quadr-candidates}).
  The upper bound on $|\candidatesmultiset|$ then follows from the following calculation:
  \begin{align*}
    |\candidatesmultiset|
       &\wwrel=    \sum_{i \,\in\, \Ico} \bigl[\, \allpieces(L_i, \underline{l})
                                     - \allpieces(L_i, \overline{l})
                                     + 1 \,\bigr] \\
       &\wwrel=    |\Ico| + \sum_{i \,\in\, \Ico} \biggl\lceil \frac{L_i}{\underline{l}} \biggr\rceil 
                   - \sum_{i \,\in\, \Ico} \biggl\lceil \frac{L_i}{\overline{l}} \biggr\rceil \\
       &\wwrel\leq |\Ico| + \sum_{i \,\in\, \Ico} \biggl[ \frac{L_i}{\underline{l}} + 1 \biggr]
                   - \sum_{i \,\in\, \Ico} \frac{L_i}{\overline{l}} \\
       &\wwrel=    |\Ico| + \Sigma_{\Ico} \cdot \frac{k + |\Ico|}{\Sigma_{\Ico}} + |\Ico|
                   - \Sigma_{\Ico} \cdot \frac{k}{\Sigma_{\Ico}} \\
       &\wwrel=    3 \cdot |\Ico| .
  \end{align*}
  A similar calculation shows the lower bound $|\candidatesmultiset| \geq |\Ico|$.
\end{proof}

If we use $f_u \!\rel= i \mapsto \lfloor \nicefrac{L_i}{\underline{l}} \rfloor$,
the candidate set is even smaller, namely $|\candidatesmultiset| \le 2 |\Ico|$.

\section{Conclusion}
\label{sec:conclusion}
We have given a formal definition of \optproblem, derived means to restrict
the search for an optimal solution to a small, discrete space of candidates,
and developed algorithms that perform this search efficiently. \wref{tab:runtime-table}
summarizes the asymptotic runtimes of the combinations of candidate space
and algorithm.
\begin{table}
  \begin{center}
    \begin{tabular}{lcc}
      \toprule
       $(f_l,f_u)$ & $\alg{f_l,f_u}$ & $\coolalg{f_l,f_u}$ \\
      \midrule
       $(1,k)$ & 
         $\Th(k n \log k)$ &
         $\Th(k n)$ \\[1ex]
       $\bigl( 1, \lceil \nicefrac{k}{i} \rceil \bigr)$\footnotemark & 
         $\Th(k \log(k) \log(n))$ &
         $\Th(k \log n)$ \\[1ex]
       $\bigl( \allpieces(L_i,\overline{l}), \allpieces(L_i,\underline{l}) \bigr)$ & 
         $\Th(n \log n)$ &
         $\Th(n)$\\
      \bottomrule
    \end{tabular}
  \end{center}
  \caption{Assuming $k \geq n$, the table shows the worst-case runtime bounds 
    shown above for the combinations of algorithm and bounding functions.}
  \label{tab:runtime-table}
\end{table}
\footnotetext{%
  The additional time $\Th(|\Ico| \log |\Ico|)$ necessary for sorting 
  $\Ico$ as required by \wref{lem:linearithmic-candidates} is
  always dominated by generating \candidatesmultiset.
}

All in all, we have shown the following complexity bounds on our problem.
\begin{corollary}\label{cor:optproblem-complexity-bounds}
  \optproblem can be solved in time and space $\Oh(n)$.
\end{corollary}
\begin{proof}
  Algorithm $\coolalg{\allpieces(L_i,\overline{l}), \allpieces(L_i,\underline{l})}$
  serves as a witness via \wref{thm:linear-alg}, \wref{lem:smartbounds-admissible}
  and \wref{lem:linear-candidates}.
\end{proof}
A simple adversary argument shows that a sublinear algorithm is impossible;
since the input is not sorted, adding a sufficiently large stick breaks any
algorithm that does not consider all sticks.
We have thus found an asymptotically optimal algorithm. 

Given its easy structure and elementary nature~-- we need but two
calls to a selection algorithm, and in fact just a single one for the typical case $k\ge n$~-- 
our method is also hard to beat in practice 
(as reported in the introduction and shown in \cite{ReitzigWild2015} 
for the apportionment variant of the algorithm).

\section*{Acknowledgments}

Erel Segal-Halevi\footnote{\url{http://cs.stackexchange.com/users/1342/}}
posed the original question~\cite{csse-question} on Computer Science Stack Exchange.
Our approach is based on observations in the answers by
Abhishek Bansal (user1990169\footnote{\url{http://cs.stackexchange.com/users/19311/}}),
InstructedA\footnote{\url{http://cs.stackexchange.com/users/20169/}} and
FrankW\footnote{\url{http://cs.stackexchange.com/users/13022/}}.
Hence, even though the eventual algorithm and its presentation have been developed 
and refined offline with the use of a blackboard and lots of paper, 
the result has been the product of a small ``crowd'' collaboration made possible 
by the Stack Exchange platform.

We thank Chao Xu for pointing us towards the work by \textcite{Cheng2014}, 
and for providing the observation we utilize in \wref{sec:selection-algorithm}.

\printbibliography

\clearpage
\appendix
\section{Notation Index}
\label{app:notations}

In this section, we collect the notation used in this paper.
Some might be seen as ``standard'', but we think
including them here hurts less than a potential 
misunderstanding caused by omitting them.

\subsection*{Generic Mathematical Notation}
\begin{notations}
\notation{$[1..n]$}
	The set $\{1, \dots, n\} \subseteq \N$.
\notation{$\lfloor x\rfloor$, $\lceil x\rceil$}
	floor and ceiling functions, as used in \cite{ConcreteMathematics}.
\notation{$\ln n$}
	natural logarithm.
\notation{$\log^2 n$}
  $(\log n)^2$
\notation{$\harm{n}$}
	$n$th harmonic number; $\harm n = \sum_{i=1}^n 1/i$.
\notation{$p_n$}
  $n$th prime number.
\notation{$\mset{A}$}
  multisets are denoted by bold capital letters.
\notation{$\mset{A}(x)$}
	multiplicity of $x$ in $\mset{A}$, i.\,e., we are using the function notation of multisets here.
\notation{$\mset{A} \uplus \mset{B}$}
	multiset union; multiplicities add up.
\notation{$\mset{L}^{(k)}$}
	The $k$th~largest element of multiset $\mset{L}$ (assuming it exists);\\
	if the elements of $\mset{L}$ can be written in non"=increasing order, 
	$\mset{L}$ is given by 
	  $\mset{L}^{(1)} \ge \mset{L}^{(2)} \ge \mset{L}^{(3)} \ge \cdots$.\\[1ex]
	\textit{Example:} 
	  For $\mset{L} = \{ 10,10,8,8,8,5 \}$, we have $\mset{L}^{(1)} = \mset{L}^{(2)} = 10$,
	  $\mset{L}^{(3)} = \mset{L}^{(4)} = \mset{L}^{(5)} = 8$ and $\mset{L}^{(6)} = 5$.
\end{notations}

\subsection*{Notation Specific to the Problem}
\begin{notations}
\notation{stick}
	one of the lengths of the input, before any cutting.
\notation{piece}
	one of the lengths after cutting; each piece results
	from one input stick after some cutting operations.
\notation{maximal piece}
	piece of maximal length (after cutting).
\notation{$n$}
	number of sticks in the input.
\notation{$\mset L$, $L_i$, $L$}
	$\mset L = \{L_1,\ldots,L_n\}$ with $L_i \in \Q_{>0}$ for all $i \in [1..n]$ contains 
	the lengths of the sticks in the input.\\
	We use $L$ as a free variable that represents (bounds on) input stick lengths.
\notation{$k$}
	$k\in\N$, the number of maximal pieces required.
\notation{$l$}
	free variable that represents (bounds on) candidate cut lengths;
	by \wref{lem:monotonicity+step} only $l = \nicefrac{L_i}{j}$ for $j\in\N$ have to be considered.
\notation{$\optlength$}
	the optimal cut length, i.\,e., the cut length that yields
	at least $k$ maximal pieces while minimizing the total
	length of non-maximal (i.\,e.\ waste) pieces.
\notation{$\cuts(L,l)$}
	the number of cuts needed to cut stick $L$ into pieces of lengths $\le l$;
	$\cuts(L,l) = \lceil \frac Ll-1 \rceil$.
\notation{$\lpieces(L,l)$}
	the number of maximal pieces obtainable by cutting stick $L$ into pieces of lengths $\le l$;
	$\lpieces(L,l) = \lfloor \frac Ll \rfloor$.
\notation{$\allpieces(L,l)$}
	the minimal total number of pieces resulting from 
	cutting stick $L$ into pieces of lengths $\le l$;
	$\allpieces(L,l) = \lceil \frac Ll \rceil$.
\notation{$\Cuts(l) = \Cuts(\mset L, l)$}
	total number of cuts needed to cut all sticks into pieces of lengths $\le l$;
	$\Cuts(\mset L,l) = \sum_{L \in \mset L} \cuts(L,l)$.
\notation{$\Lpieces(l) = \Lpieces(\mset L, l)$}
	total number of maximal pieces resulting from cutting stick $L$ into pieces of lengths $\le l$;
	$\Lpieces(\mset L,l) = \sum_{L \in \mset L} \lpieces(L,l)$.
\notation{$\Feasible(l) = \Feasible(\mset L, k, l)$}
	indicator function that is $1$ when $l$ is a feasible length and $0$ otherwise;
	$\Feasible(\mset L, k, l) = [m(\mset L,l) \ge k]$.
\notation{$\candidatesmultiset(I,f_l, f_u)$}
	multiset of candidate lengths $\nicefrac{L_i}{j}$, restricted by the index 
	set~$I$ of considered input sticks $L_i$, and lower resp.\ upper bound on $j$;
	cf.\ \wpref{def:candidates}.
\notation{$\allcandidatesmultiset$}
	The unrestricted (infinite) candidate set 
	$\allcandidatesmultiset = \candidatesmultiset([1..n],1,\infty)$.
\notation{admissible restriction $(I,f_l,f_u)$}
	sufficient conditions on restriction $(I,f_l,f_u)$ to ensure that 
	$\optproblem \in \candidates(I,f_l,f_u)$;
	cf.\ \wpref{lem:admissible-bounds}.
\notation{$I_{>L}$}
  the set of indices of input sticks $L_i > L$; cf.\ \wref{lem:feasible-computation}.
\notation{$\Ico$, $\Lco$}
	$\Ico = I_{> \Lco}$ is our canonical index set with cutoff length 
	$\Lco$ the $k$th~largest input length; 
	cf.\ \wpref{def:ourindexset}.
\notation{$\Sigma_I$}
  assuming $I \subseteq [1..n]$, this is a shorthand for $\sum_{i \in I} L_i$.
\notation{$\underline l$, $\overline l$}
	lower and upper bounds on candidate lengths $\underline l \le l\le \overline l$
	so that $\optlength \in \allcandidates \cap [\underline l, \overline l]$;
	see \wpref{lem:smartbounds-admissible}.
\end{notations}

\section{On the Number of Distinct Candidates}
\label{app:additional-results}

As mentioned in \wref{sec:algorithms}, algorithm \alg{} can profit from removing
duplicates from the candidate multisets during sorting. We will show in the
subsequent proofs that none of the restrictions introduced above cause more than
a constant fraction of all candidates to be duplicates.

We denote with $\candidates(\dots)$ the set obtained by removing duplicates
from the multiset $\candidatesmultiset(\dots)$ with the same restrictions. 

\begin{lemma}\label{lem:quadr-candidates-wc}
  $|\candidates(\Ico,1,k)| \in \Th\bigl( |\candidatesmultiset(\Ico,1,k)| \bigr)$ 
  in the worst case.
\end{lemma}
\begin{proof}
  Let for short $C \ce |\candidates(\Ico, 1, k)|$ and 
  $\candidatesmultiset \ce \candidatesmultiset(\Ico,1,k)$. 
  It is clear that $C \leq |\candidatesmultiset|$; 
  we will show now that $C \in \Om(|\candidatesmultiset|)$ in the worst case.
  
  Consider instance
    \[ \mset{L}_{\mathrm{primes}} = \{p_n, \dots, p_1\} \]
  with $p_i$ the $i$th prime number and any $k \in \N$; note that $L_i = p_{n-i+1}$.
  Let for ease of notation $n' \ce \min(k,n+1)$; 
  note that $\Lco = \mset{L}_{\mathrm{primes}}^{(n')}$ if $k \leq n$.
  We have $|\Ico| = \min(k-1,n)$ because the $L_i$ are 
  pairwise distinct, and therefore $|\candidatesmultiset| = k|\Ico| = k(n'-1)$. 
  Since the $L_i$ are also pairwise coprime, all candidates 
  $\nicefrac{L_i}{j}$ for which $j$ is \emph{not} a multiple of $L_i$
  are pairwise distinct. 
  Therefore, we get 
  \begin{align*}
    C &\geq |\candidatesmultiset| - \sum_{i=n-n'+2}^n \Biggl\lfloor \frac{k}{p_i} \Biggr\rfloor \\
      &\geq |\candidatesmultiset| - \sum_{i=n-n'+2}^n \frac{k}{p_i} \\
      &=    |\candidatesmultiset| - (n'-1)k \cdot \sum_{i=n-n'+2}^n \frac{1}{p_i} \\
      &\geq |\candidatesmultiset| - |\candidatesmultiset| \cdot \frac{n'}{p_{n'}} \\
      &\geq |\candidatesmultiset| - |\candidatesmultiset| \cdot \frac{2}{3} \\
      &=    \frac{|\candidatesmultiset|}{3}.
  \end{align*}
  In particular, we can show that $\nicefrac{k}{p_k} \leq \nicefrac{2}{3}$ by 
  $\nicefrac{k}{p_k} < 0.4$ for $k \geq 20$ \cite[eq.\ (4.20)]{ConcreteMathematics} 
  and checking all $k < 20$ manually; the maximum is attained at $k=2$.
\end{proof}

\begin{lemma}\label{lem:linearithmic-candidates-wc}
  $|\candidates(\Ico,1,\lceil \nicefrac{k}{i} \rceil)| 
    \in \Th\bigl( |\candidatesmultiset(\Ico,1,\lceil \nicefrac{k}{i} \rceil)| \bigr)$ 
  in the worst case.
\end{lemma}
\begin{proof}
  Let for short $C \ce |\candidates(\Ico, 1, \lceil \nicefrac{k}{i} \rceil)|$ and 
  $\candidatesmultiset \ce \candidatesmultiset(\Ico,1,\lceil \nicefrac{k}{i} \rceil)$. 
  It is clear that $C \leq |\candidatesmultiset|$; 
  we will show now that $C \in \Om(|\candidatesmultiset|)$ in the worst case.

  We make use of the same instance $(\mset{L}_{\mathrm{primes}}, k)$ 
  we used in the proof of \wref{lem:quadr-candidates-wc},
  with a similar calculation:
  \begin{align*}
    C &=    |\candidatesmultiset| - \sum_{i=n-n'+2}^n \Biggl\lfloor 
                     \frac{\lceil \nicefrac{k}{i} \rceil}{p_i} 
                   \Biggr\rfloor \\
      &\geq |\candidatesmultiset| - \sum_{i=n-n'+2}^n \frac{\frac{k}{i} + 1}{p_i} \\
      &\geq |\candidatesmultiset| - \frac{n'-1}{p_{n'-1}} \cdot \biggl(1 + \frac{k}{n'-1}\biggr) \\
      &\geq |\candidatesmultiset| - \frac{2}{3} \cdot \biggl( 1 + \frac{k}{n'-1} \biggr) \\
      &\in \Th(|\candidatesmultiset|)
  \end{align*}
  because $\nicefrac{k}{n'} \in \oh(k \log n') = \oh(|\candidatesmultiset|)$.
\end{proof}

\begin{lemma}%
  \label{lem:linear-candidates-wc}%
  $|\candidates(\Ico,\allpieces(L_i,\overline{l}), \allpieces(L_i,\underline{l}))| 
    \in \Th\bigl( |\candidatesmultiset(\Ico,\allpieces(L_i,\overline{l}), \allpieces(L_i,\underline{l}))| \bigr)$ 
  in the worst case.
\end{lemma}
\begin{proof}
  Let again for short
    $C \ce |\candidates(\Ico,\allpieces(L_i,\overline{l}), \allpieces(L_i,\underline{l}))|$ 
  and 
    $\candidatesmultiset \ce 
       \candidatesmultiset(\Ico,\allpieces(L_i,\overline{l}), \allpieces(L_i,\underline{l}))$.
  It is clear that $C \leq |\candidatesmultiset|$; 
  we will show now that $C \in \Om(|\candidatesmultiset|)$ in the worst case.
   
  We make use of our trusted instance $(\mset{L}_{\mathrm{primes}}, k)$ again. 
  We show that very prime yields at least one candidate unique to itself, 
  as long as $k$ is constant (which is sufficient for a worst"=case argument).
   
  Recall that $\overline{l} > \underline{l}$ so every $L_i$ \emph{has} some
  $j$; we note furthermore that for fixed $i \in \Ico$,
    \[            j 
       \wwrel\leq \allpieces(L_i, \underline{l})
       \wwrel=    \left\lceil L_i \bigg/%
                                     \frac{\Sigma_{\Ico}}%
                                          {k + |\Ico|}  
                  \right\rceil
       \wwrel=    \left\lceil p_i \cdot %
                                     \frac{k + |\Ico|}%
                                          {\sum_{i'\in \Ico} p_{i'}}  
                  \right\rceil
       \wwrel\leq \left\lceil p_i \cdot %
                       \frac{2k}%
                            {p_n}               
                        \right\rceil 
       \wwrel\leq 2k
       \wwrel<    L_i \]
  for big enough $n$, in particular because 
  $p_n \sim n \ln n$~\cite[p\,110]{ConcreteMathematics}.
  That is, every $L_i$ with $i \in \Ico$ yields at least one $\nicefrac{L_i}{j}$ 
  no other does, since all $L_i$ are co"=prime.
  Hence $C \geq |\Ico| \in \Th(|\candidatesmultiset|)$.
\end{proof}

\end{document}